\documentclass[12pt]{article}
\usepackage{amssymb}
\usepackage[utf8]{inputenc}
\usepackage{amsthm}
\usepackage{cite} 
\usepackage[english]{babel}
\usepackage{hyperref}
\usepackage{xcolor}
\usepackage{bbold} 
\newtheorem{proposition}{Proposition}

\renewcommand{\theequation}{\arabic{equation}}
\newcommand{\be}{\begin{equation}}
\newcommand{\ee}{\end{equation}}
\newcommand{\bea}{\begin{array}}
	\newcommand{\ea}{\end{array}}
\newcommand{\beqa}{\begin{eqnarray}}
\newcommand{\eeqa}{\end{eqnarray}}
\newcommand{\bean}{\begin{eqnarray*}}
	\newcommand{\eean}{\end{eqnarray*}}

\def\up#1{\leavevmode \raise.16ex\hbox{#1}}

\newcommand{\gapproxeq}{\lower
	.7ex\hbox{$\;\stackrel{\textstyle >}{\sim}\;$}}
\newcommand{\lapproxeq}{\lower .7ex\hbox{$\;\stackrel
		{\textstyle <}{\sim}\;$}}
\renewcommand{\theequation}{\thesection.\arabic{equation}}
\newcounter{appendice}
\newcommand{\appendice}
{
	\setcounter{equation}{0}
	\renewcommand{\theequation}{\Alph{appendice}.\arabic{equation}}
	\addtocounter{appendice}{1}
	{\Large{\bf Appendix \Alph{appendice}}}
}
\def\thebibliography#1{{\bf REFERENCES\markboth
		{REFERENCES}{REFERENCES}}\list
	{[\arabic{enumi}]}{\settowidth\labelwidth{[#1]}\leftmargin\labelwidth
		\advance\leftmargin\labelsep
		\usecounter{enumi}}
	\def\newblock{\hskip .11em plus .33em minus -.07em}
	\sloppy
	\sfcode`\.=1000\relax}

\def\z{{\zeta}}
\def\BI{{\mathbb 1}} 

\definecolor{patriarch}{rgb}{0.9, 0.0, 0.5}

\begin{document}

\centerline{\LARGE Exact solutions for scalars and spinors  on quantized  }
\vskip .25cm
\centerline{\LARGE   (Euclidean) $AdS_2$ and the  correspondence principle}
\vskip 1cm


\centerline{ A. Pinzul${}^1$\footnote{apinzul@unb.br} and A. Stern$^{2}$\footnote{astern@ua.edu}}

\vskip 1cm
\begin{center}

  {${}^1$ Universidade  de Bras\'{\i}lia, Instituto de F\'{\i}sica\\
70910-900, Bras\'{\i}lia, DF, Brasil\\
and\\
International Center of Physics\\
C.P. 04667, Bras\'{\i}lia, DF, Brazil
\\}

%
%
%
  {${}^2$ Department of Physics, University of Alabama,\\ Tuscaloosa,
Alabama 35487, USA\\}

\end{center}

\vskip 0.5cm

\abstract{ We obtain  the exact solutions to the field equations for massless scalar and massless spinor fields on quantized two-dimensional anti-de Sitter space. We then apply the  $AdS/CFT$ correspondence principle to  get exact answers for the two point correlation functions  of the  associated   operators on the boundary. The   results  support the conclusion that  conformal symmetry on the boundary is not spoiled by  quantization of the bulk.  Moreover,  quantization of the bulk has no effect on the   spinor two-point correlation function, while  it induces an overall re-scaling in the scalar two-point correlation function.

\newpage}
\section{Introduction}

Research on the quantization of space-time has been carried out with the intention of probing the quasiclassical regime of quantum gravity.  In this regard,  arguments have been given supporting the  idea that  the quasiclassical regime  can be  described in terms of quantum field theory on a non-commutative, i.e., quantized, background.\cite{Doplicher:1994tu}   Meanwhile, the $AdS/CFT$ correspondence principle offers the possibility  of   a nonperturbative description of quantum gravity on asymptotically AdS spaces.\cite{Maldacena:1997re,Witten:1998qj,Freedman:1998tz,Beisert:2010jr} (Despite this, many  applications of the correspondence principle rely on {\it classical} gravity calculations.)
In an attempt to merge these different approaches, we have   studied  the  $AdS/CFT$ correspondence principle in a non-commutative setting.\cite{Pinzul:2017wch,deAlmeida:2019awj ,Lizzi:2020cwx} By this we mean that the geometry on the gravity side of the correspondence was written on a quantized space-time, which we briefly review below.

In \cite{Pinzul:2017wch,deAlmeida:2019awj} we  examined the example where the bulk geometry is the  quantized version of two-dimensional anti-de Sitter space ($AdS_2$).  In this regard, $AdS_2$ has a unique quantization upon imposing the requirement that all the  isometries of the classical geometry are preserved under the quantization.\cite{Ho:2000fy,Fakhri:2003cu,Fakhri:2011zz,Jurman:2013ota,Stern:2014aqa,Chaney:2015ktw}
 \footnote{Although it may not be possible to make the same  demand for arbitrary $d>2$ dimensional anti-de Sitter spaces, higher-dimensional  generalizations of such an isometry preserving quantization procedure do exist for other spaces, namely,  indefinite complex projective spaces  ($\mathbb{CP}^{p,q}$, $q\ge 1$), as was shown in \cite{Lizzi:2020cwx}.}
 The $AdS/CFT$ correspondence principle posits a strong/weak duality between quantum gravity in some $d$ dimensional bulk and a field theory on its $d-1$ boundary.\cite{Maldacena:1997re,Witten:1998qj,Freedman:1998tz,Beisert:2010jr}  The bulk is traditionally taken to be an asymptotically  anti-de Sitter  space, with a corresponding  conformal field theory (CFT) residing on its boundary, although many extensions of the correspondence having been proposed.
 We  argued  in \cite{Pinzul:2017wch,deAlmeida:2019awj}  that the quantized version of $ AdS_2$ is an asymptotically  anti-de Sitter  space.  This is due to i) the result that the star product used to realize the noncommutative algebra approaches the trivial commutative product in the boundary limit, and ii) all quantum corrections to the Killing vectors vanish in this limit.  \footnote{Analogous results were shown for the quantization of $\mathbb{CP}^{p,q},\;q\ge 1$, implying that  these spaces are asymptotically commutative.} Therefore, according to the  correspondence principle,  a dual conformal quantum theory should reside on the boundary of quantized $ AdS_2$ (baring  known difficulties of the correspondence principle for two dimensional anti-de Sitter space\cite{Strominger:1998yg,Maldacena:2016hyu,Jensen:2016pah,Mezei:2017kmw}).   An initial study of the  corresponding boundary theory was  pursued in \cite{Pinzul:2017wch,deAlmeida:2019awj}.  There we examined the two and three point correlation functions of operators on the boundary which were sourced by scalar fields on quantized $AdS_2$.   [More precisely, the analysis and quantization was carried out for the Euclidean version of $AdS_2$ ($EAdS_2$).] Due to rather involved calculations in the approach followed there, the best we could do was obtain the leading order quantum (non-commutative) corrections to the  usual two and three point correlation functions, i.e., those associated with a `classical' space-time bulk.

 In this article we re-examine the correspondence principle for quantized  $EAdS_2$, and obtain {\it exact} results for the two point correlation functions  of operators on the boundary. This is done for operators sourced by both scalar and spinor fields in the bulk.  The exact results for the two point correlators are possible due to the existence of exact solutions to both the scalar and spinor field equations on quantized $ EAdS_2$.  These equations,   which involve the relevant  Laplacian and Dirac operators, can be written purely algebraically.
 For simplicity, we shall specialize here to  the case of massless scalar and spinor fields, although exact solutions should be obtainable for the massive case, as well.  Our result shows that  the conformal properties of the boundary two-point correlation functions are not spoiled by the quantization of the bulk.  This is not unexpected since the corresponding isometries on the gravity side are preserved.  Moreover, the boundary two-point correlation function sourced by spinors on quantized $EAdS_2$ is {\it identical} to that of its classical counterpart.  Therefore, from such  boundary correlation functions, one cannot  distinguish whether the bulk theory is classical or quantum.   On the other hand, the boundary two-point correlation function sourced by scalars on quantized $EAdS_2$ differs from its classical counterpart by a simple overall factor, which agrees with the leading order result of \cite{Pinzul:2017wch,deAlmeida:2019awj}.

The outline for the rest of this article is as follows:  In section 2 we review the dynamics of massless scalar and spinor fields on $EAdS_2$.  Two formulations of spinors are given, which we refer to as the noncovariant and covariant formulations, and a map between the two is given.  While the noncovariant approach is simpler, the covariant approach has a natural generalization to quantum space-time.  As a preliminary step towards quantization we introduce a Poisson structure to  $EAdS_2$.
The quantization of  $EAdS_2$ is reviewed in section 3.  There we write down the Laplacian and Dirac operators on quantized $EAdS_2$, and present exact solutions to the equations for both massless scalar and  spinor fields on this space.  From the asymptotic properties of these solutions, we derive exact results for the associated two-point correlation functions on the boundary in section 4.  Concluding remarks are presented in section 5. In appendix A we show how to evaluate boundary terms on quantized $EAdS_2$.  We find the quantum analogue of the map between the covariant and noncovariant bases  in appendix B, and use it to construct the Dirac operator in the noncovariant basis. { The solutions for massless spinors in that basis are trivially obtained.  They can then be  mapped back to the covariant basis, and the result agrees with the exact solutions of section 3. }

\section{Scalars and spinors on $EAdS_2$}
\setcounter{equation}{0}

\subsection{Preliminaries}

We begin with a review of the Euclidean $AdS_2$ geometry.
 $EAdS_2$ is conventionally coordinatized using Fefferman-Graham coordinates, $(z,t)$, spanning a half-plane ${\mathbb{R}}^{2}_+$, $z>0$, $ - \infty<t<\infty$.  The asymptotic boundary corresponds to $z\rightarrow 0$. The metric tensor in terms of these coordinates is given by
\be ds^2=\frac{dz^2+dt^2}{z^2}\;.\ee
$EAdS_2$ is a maximally isometric space, with the isometries generated by the three Killing vectors
\be
{\cal K}^-=-\partial_t\quad ,\quad{\cal K}^1=-t\partial_t-z\partial_z\quad ,\quad {\cal K}^+=(z^2- t ^2)
\,\partial_t -2 z t \,\partial_z \;,\label{undfrmdKlng}
\ee
which form a basis for the $su(1,1)$ algebra.  Defining ${\cal K}^0=\frac 12({\cal K}^++{\cal K}^-)$ and  ${\cal K}^2=\frac 12({\cal K}^+-{\cal K}^-)$, we get the commutator brackets
\be [{\cal K}^0,{\cal K}^1]={\cal K}^2\;,\qquad\; [{\cal K}^1,{\cal K}^2]=-{\cal K}^0\;,\qquad\; [{\cal K}^2,{\cal K}^0]={\cal K}^1\;.\qquad\; \label{Kcmtrs}\ee

The Laplacian on  $EAdS_2$ is
\be\Delta=z^2(\partial_z^2+\partial_t^2)\;.  \ee
It can be written in terms of the three Killing vectors
\be\Delta={\cal K}^a{\cal K}_a
\;,  \ee
where the index $a=0,1,2$ is raised and lowered using the metric tensor
$[\eta_{ab}]=$ diag$(-1,+1,+1)$.
The equation for a free massless  scalar $\Phi$ on $EAdS_2$,
\be \Delta\Phi=0\;,\label{mslssclrfe}\ee
 is  easily solved.
General solutions  are expressed in terms of analytic and anti-analytic functions of  $\zeta=z+it$,
\be
\Phi(\zeta,\bar\zeta)={f(\zeta)+  g(\bar \zeta)}\;,\label{sfmslssln}
\ee
$f$ and $g$  being complex functions.  Upon approaching the boundary $z=0$ the solution behaves as
$$ \Phi(\zeta,\bar\zeta)=\phi_0(t)+ z \phi_1(t) +{\cal O}(z^2)\;,$$
\be \phi_0(t)={f(it)+  g(-it)}\;,\qquad\; \phi_1(t)=\Bigl({f'(\zeta)+  g'(\bar\zeta)}\Bigr)|_{z=0}\;,\label{bhmslssln}\ee
the prime denoting a derivative.

The Dirac operator on an  arbitrary $d-$dimensional commutative manifold is
\be D=\gamma^\mu(\partial_\mu+\omega_\mu)\;,\label{theeDop}\ee
where $\omega_\mu$, $\mu=1,...,d$ are the components of the spin connection one form and our convention is that the gamma matrices $\gamma^\mu$ satisfy the anti-commutator
\be [\gamma^\mu,\gamma^\nu]_+=2g^{\mu\nu}\BI\;.\label{gmanticmtr}\ee
Upon assuming that  the manifold  is locally Euclidean, the gamma matrices are transformed to a local Euclidean frame using vielbein fields $e^i_{\;\;\mu}$:   $\gamma^\mu\rightarrow \hat \gamma^i=e^i_{\;\;\mu}\gamma^\mu$.  The transformed gamma matrices  $ \hat \gamma^i$   satisfy the anti-commutation relations: $[\hat \gamma^i,\hat \gamma^j]_+=2\delta^{ij}\BI$.  We shall assume that the space has zero torsion.   Denoting a basis of one forms by $\{dx^\mu\}$, we write down  the vielbein and spin connection one forms $e=e_{\mu}dx^\mu$
and $\omega=\omega_\mu dx^\mu$, respectively.  $e_\mu$ and $\omega_\mu $ take values in the Clifford algebra generated by $\hat \gamma^i$,   $e_\mu=e^i_{\;\;\mu}\hat \gamma_i$ and $\omega_\mu=\frac 18\omega_{ij\mu}[\hat \gamma^i,\hat \gamma^j]$.
Then the zero torsion condition is
\be de+\omega\wedge e+ e\wedge\omega=0\;.\label{trsnfree}\ee

In the next two subsections we shall write down two different realizations of the Dirac operator, one in what we call the noncovariant basis and another in the covariant basis. While the former basis is more familiar, the latter is crucial for the isometry preserving quantization.    We also write down the map between the two realizations.  A relevant feature of this map is that it is singular at the asymptotic boundary.  A preliminary step towards quantization is taken in the third subsection by introducing a  natural Poisson structure on the manifold.

\subsection{Dirac operator in a noncovariant basis}
We denote gamma matrices, the Dirac operator  and spinors in the noncovariant basis  with a tilde.  (We shall drop the tilde in the next subsection for the corresponding quantities in the covariant basis.)
In writing down  the Dirac operator on  $EAdS_2$ we can take the two gamma matrices $\{\gamma^\mu\}=\{ \tilde\gamma^{ z}, \tilde \gamma^{ t}\}$ to be
\be  \tilde\gamma^{ z}=z\sigma_1\;,\qquad\quad \tilde \gamma^{ t}= z\sigma_2\;, \ee
$\sigma_i$ being Pauli matrices.
This choice satisfies (\ref{gmanticmtr}).   The vielbein one-form is then
\be e=\sigma_i e^i_{\;\;\mu}dx^\mu=\frac{1}{z}(\sigma_1 dz+\sigma_2 dt) \;.\ee
The torsion free condition (\ref{trsnfree}) is satisfied for
	\be  \omega=\frac{ i}{2z}\sigma_3  dt\;.\ee
The Dirac operator (\ref{theeDop}) then becomes
\beqa{\tilde  D}&=&z(\sigma_1\partial_z+\sigma_2\partial_t)-\frac 12\sigma_1\;=\;z^{3/2}(\sigma_1\partial_z+\sigma_2\partial_t)\frac 1{\sqrt{z}} \cr&&\cr &=&z^{3/2}\pmatrix{0&\partial_z-i\partial_t\cr\partial_z+i\partial_t &0}\frac 1{\sqrt{z}}\;.\label{cmttvdo} \eeqa
The corresponding chirality operator is $\sigma_3$ as it anti-commutes with the Dirac operator $[{\tilde D},\sigma_3]_+=0$,  and squares to the identity.
 ${\tilde D}$ satisfies the Lichnerowicz formula :
\be\sqrt{z}\,{\tilde D}\frac 1z\,{\tilde  D}\,\sqrt{z}=\Delta\;.\label{Dsquared}\ee

The free massless  spinor ${\tilde \Psi}$ on $EAdS_2$ satisfies
$ {\tilde D\tilde\Psi}=0$,
which is  easily solved.
General solutions  are again expressed in terms of analytic and anti-analytic functions of  $\zeta=z+it$,
\be {\tilde \Psi}(\zeta,\bar\zeta)=\sqrt{z}\pmatrix{f(\zeta)\cr  g(\bar \zeta)}\;,\qquad z= Re\, \zeta\;.\label{mslssln}\ee
Here the complex functions $f$ and $g$  associated with solutions of different chirality.  In the boundary limit
 $z\rightarrow 0$,
\be {\tilde \Psi}(\zeta,\bar\zeta)\rightarrow\sqrt{z}\,\tilde\psi_0(t)\;,\qquad\quad\tilde \psi_0(t)=\pmatrix{f(it)\cr  g(-it)}\;.\label{bndrylmtps}\ee

\subsection{Dirac operator in the covariant basis}

In formulating the covariant form of the Dirac operator, it is convenient to embed $EAdS_2$ in  ${\mathbb{R}}^{1,2}$.  Calling the embedding coordinates $X^a$, $a=0,1,2$, $EAdS_2$ corresponds to a single sheeted hyperboloid
\be X^a X_a=-1 \label{XaXaismnsn}\;,
\ee
where  the metric on $ {\mathbb{R}}^{1,2}$ is $[\eta_{ab}]=$diag$(-1,+1,+1)$, and, for convenience, the $AdS_2$ scale is normalized to one.\footnote{Our index convention here is different  from that in
\cite{Pinzul:2017wch,deAlmeida:2019awj}. In the latter, $a=2$ corresponded to the time-like direction.}
 Expressed in terms of Fefferman-Graham coordinates the embedding is given by
\be
X^0=-\frac {z^2+t^2+1}{2z}\;,\qquad\;  X^1=-\frac tz\;, \qquad\;X^2=-\frac {z^2+t^2-1}{2z} \label{XsnFGcrds}\;,
\ee
putting the single sheeted hyperboloid in the region $X^0<0$.

In the covariant formulation of the Dirac operator we choose the chirality operator to be
\be
\gamma =X^a\tau_a\;,\label{cvrntcrltep}
\ee
where $\{\tau_a,\, a=0,1,2\}$ is a basis  for  $su(1,1)$ in the defining representation.  Our choice is  $ (\tau_0,\tau_1,\tau_2)=(\sigma_3,i\sigma_1,i\sigma_2)\,,$ $\sigma_a$ denoting Pauli matrices,  which satisfies
\be  \tau_a\tau_b=-i\epsilon_{abc}\tau^c-\eta_{ab}\BI\;,\label{towprptees}\ee
 $\epsilon_{abc}$ being the Levi-Civita symbol with $\epsilon_{012}=-1$. We note that $\gamma $ is not hermitean, and its square is the unit operator.

The  chirality operator $ \gamma $ can be mapped from the  chirality operator in the noncovariant basis, $\sigma_3$, using the similarity transformation
\be \gamma =U\sigma_3U^{-1}\;,\ee
where  the $2\times 2$ matrix $U$ expressed in complex coordinates $\zeta$ and $\bar \zeta$ is given by
\be  U=\frac 1{2\sqrt{z}}\pmatrix{-\zeta+1&-\bar \zeta -1\cr \zeta+1&\bar\zeta -1}\;.
\label{Uoprtr}\ee
$U$ has determinant equal to one and satisfies
\be U^\dagger\sigma_3U=-\sigma_3\;.\label{prptee4U}\ee
While $U$ has unit determinant for all $z$, $U$   is  singular   in the boundary limit, $z\rightarrow 0$.  More specifically, det$\;\sqrt{z}U\rightarrow 0$ in the limit.

If we apply the same  similarity transformation to the Dirac operator written in the noncovariant basis (\ref{cmttvdo}),
\be  {\tilde D}\rightarrow D=U{\tilde D}U^{-1}\;,\label{tmbtwnDs}\ee we get
\be   D\,=\,i\,\pmatrix {tz\partial_z+\frac 12 (1+t^2-z^2)\partial_t+i &z(t+i)\partial_z-\frac 12\Bigl(z^2-(t+i)^2\Bigr)\partial_t\cr -z(t-i)\partial_z+\frac 12\Bigl(z^2-(t-i)^2\Bigr)\partial_t& -tz\partial_z-\frac 12 (1+t^2-z^2)\partial_t+i}\;.\label{Dprm}\ee
In comparing with (\ref{theeDop}), it corresponds to making the following choice for the gamma matrices
\be  \gamma^z\,=\,iz\pmatrix {t &t+i\cr -t+i& -t}\qquad\quad   \gamma^t\,=\,\frac i2 \pmatrix {1+t^2-z^2 &-z^2+(t+i)^2\cr z^2-(t-i)^2&- (1+t^2-z^2)}\;.\label{altgms}\ee It can be checked that they satisfiy  (\ref{gmanticmtr}). The gamma matrices (\ref{altgms})
 are not  obtained with a local Lorentz transformation from $\hat \gamma^i=\sigma_i$.
The transformed Dirac operator $D$ can be re-expressed in terms of the three Killing vectors (\ref{undfrmdKlng}),
\beqa  D&=&i\pmatrix{ -\frac 12({\cal K}^-+{\cal K}^+)+i&\frac 12({\cal K}^--{\cal K}^+)-i{\cal K}^1\cr -\frac 12({\cal K}^--{\cal K}^+)-i{\cal K}^1&\frac 12({\cal K}^-+{\cal K}^+) +i} \cr&&\cr&=&-i\tau_a{\cal K}^a-\BI\;.\label{Dpryme}\eeqa
This Dirac operator is preserved under simultaneous $2+1$ Lorentz transformations on $\tau_a$ and ${\cal K}^a$.
By construction, $D$ anticommutes with the chirality operator  $\gamma$. This can also be independently verified using the properties (\ref{Kcmtrs}), (\ref{towprptees}) and
\be   [{\cal K}^a, X^b]=\epsilon^{abc} X_c\;,\qquad\quad X_a{\cal K}^a= {\cal K}^aX_a=0\;.\label{cntrKXa}\ee

The solution $\Psi$ to $ D\Psi=0$ is obviously obtained by transforming  the solution (\ref{mslssln})  by $U$
\be
\Psi=U\tilde\Psi=\frac 1{2}\pmatrix{-\zeta+1&-\bar \zeta -1\cr \zeta+1&\bar\zeta -1}
\pmatrix{f(\zeta)\cr  g(\bar \zeta)}\;.\label{mslsslnprm}
\ee
The boundary behavior of $\Psi$ differs from that of $\tilde\Psi$ in (\ref{bndrylmtps}) due to the fact that $U$ is singular at the boundary.  Upon Taylor expanding  (\ref{mslsslnprm}) at $z=0$ up to first order, we get
\beqa \Psi(\zeta,\bar\zeta)&= & \frac{1}{2}
\left(
\begin{array}{c}
 1-i t \\1+ i t \\
\end{array}
\right) (f(i t)-g(-i t))\cr&&\cr&+&\frac{z}{2} \left(
\begin{array}{c}
-f(i t)-g(-i
   t)+(1-i t) \left(f'(i t)-g'(-i t)\right)\\
\;\;\, f(i t)+g(-i t)+(1+i
   t) \left(f'(i t)-g'(-i t)\right)\\
\end{array}
\right)+O\left(z^2\right)\;.\cr&& \label{nxt2ldngrdr}\eeqa
We can write the result as
\be \Psi(\zeta,\bar\zeta)=
 (\sigma_3-i t\BI )\,\psi_0(t)
- i  z\,\partial_t\Bigl(  (\sigma_3-i t\BI )\,\psi_1(t)\Bigr)+{\cal O}(z^2)\;,\label{cantoofrthree}\ee
where
\be\psi_0(t)=\frac12 \Bigl(f(it)- g(-it)\Bigr)\pmatrix{1\cr -1}\;,\quad\quad\psi_1(t)=\frac 12\Bigl(f(it)+ g(-it)\Bigr)\pmatrix{1\cr -1}\;.\label{psi01mns}\ee

\subsection{Poisson structure}

A preliminary step towards quantization is to attach a Poisson structure to the manifold.   Poisson brackets $\{\,,\,\}$  can be uniquely defined on $EAdS_2$ upon demanding that they preserve the isometry.   This choice corresponds to having the embedding coordinates $X^a$ define
 an $su(1,1)$ algebra
\be \{X^a,X^b\}=\epsilon^{abc} X_c\;.\label{PbsfXaXb}\ee
These relations are consistent with writing down the following Poisson bracket for the Feffermann-Graham coordinates spanning ${\mathbb{R}}^{2}_+$\cite{Pinzul:2017wch,deAlmeida:2019awj}
\be\label{tz_class}
\{t,z\}= z^2 \;,
\ee
which can be easily verified using (\ref{XsnFGcrds}).  It follows that $z^{-1}=X^2-X^0$ and $t$ are canonically conjugate coordinates on  ${\mathbb{R}}^{2}_+$,
\be
\{z^{-1},t\}=1\;.\label{zinvtcnpr}
\ee
Alternatively, one can define a Darboux map to coordinates  $(x,y)$ spanning the entire plane, where  {
\be x=-\ln \,z\,,\qquad y=t/z\;, \label{Drbx}\ee and consequently}
$\{x,y\}=1$.

 Using (\ref{undfrmdKlng}),  the action of the Killing vectors on functions ${\cal F}$ on $EAdS_2$ is implemented from Poisson brackets with the embedding coordinates
\be {\cal K}^a{\cal F}=\{ X^a,{\cal F}\}\;.\label{KaFa}\ee
It then trivially follows that the Poisson bracket and the defining relation (\ref{XaXaismnsn}) are
 preserved under the action of the Killing vectors
\beqa  &&{\cal K}^d\{X^a,X^b\}=\{X^d, \{ X^a,X^b\}\}=\epsilon^{abc}\{X^d,  X_c\}= \epsilon^{abc}{\cal K}^dX_c\;,\cr&&\cr
&& {\cal K}^a(X^b X_b)=\{X^a,X^b X_b\}=2\{X^a,X^b\} X_b=0 \;.\eeqa

The field equations for the scalar  and spinor can be expressed using Poisson brackets with the embedding coordinates.  For the massless scalar field $\Phi$,
\be\Delta\Phi=\{ X^a,\{X_a,\Phi\}\}=0\;,\label{ctvanlgfyeq}
\ee
while for the massless spinor $\Psi$,
\be D\Psi=-i\{\gamma,\Psi\}-\Psi =0\;.\label{DeqntPBs}\ee

\section{Scalars and spinors on quantized  $EAdS_2$}
\setcounter{equation}{0}

We begin with the natural isometry preserving quantization of $EAdS_2$ in the first subsection.  There we show how the field equations for the scalar and spinor can be written in an algebraic manner. The exact solutions for massless fields are presented in the second subsection.

\subsection{Quantization}

The quantization of  $EAdS_2$  is unique upon demanding that it preseves the isometry of  $EAdS_2$.\cite{Ho:2000fy,Jurman:2013ota} For this one maps  the embedding coordinates $ X^a$ to operators $ \hat X^a$, satisfying the $su(1,1)$ algebra
\be
[\hat X^a,\hat  X^b]=i\alpha\epsilon^{abc} \hat X_{ c}\;, \label{stpdids4X}
\ee
 where  $\alpha$ is the quantization parameter.
 As with commutative $EAdS_2$, we fix the scale to be one
\be
\hat X^a\hat X_{a}=-\BI\;,\label{sclstpdids4X}
\ee
in analogy with (\ref{XaXaismnsn}). [Any quantum corrections to the Casimir can be removed with a rescaling of $\hat X_a$ and $\alpha$.]

$\hat X_a$ generate the algebra ${\cal A}$ of quantized $EAdS_2$. {They form a basis for  $su(1,1)$  in a unitary irreducible representation.  It was shown in \cite{Pinzul:2017wch} that both the commutative limit and boundary limit can be formulated in terms of these  representations.  In order for there to be a consistent commutative limit to  $EAdS_2$, one has to restrict  the allowed representations to  the  discrete series  $D^\pm(k)$, where $k<0$ labels the representation.   Upon choosing $D^+(k)$, we can go to an eigenbasis of $\hat X^0$, which we denote by $|k,m>$, $m=0,1,2...$.  The action of $\hat X^a$ on these states is
\beqa \hat X_+|k,m>&=&-\alpha  c_m|k,m+1> \;, \cr&&\cr
 \hat X_-|k,m>&=&-\alpha c_{m-1}|k,m-1> \;, \cr&&\cr
 \hat X^0|k,m>&=& -\alpha(m-k)|k,m>\;,\label{Xsonegbs}
\eeqa
where $\hat X_\pm= \hat X^2\pm i \hat X^1$ and $c_m=\sqrt{(m+1)(m-2k)}$.
For  the action of the quadratic Casimir one gets
\beqa \hat X^a\hat X_a|k,m>&=& -\alpha^2 k(k+1)|k,m>\;.\eeqa
Upon comparing with (\ref{sclstpdids4X}), we get  that  $k$  is  related to noncommutativity parameter,
 \be\alpha^2 k(k+1)=1\;,\label{kvsalfa}\ee and we should make the restriction $k<-1$.  The commutative limit $\alpha\rightarrow 0$ then corresponds to $k\rightarrow -\infty$.  In section 3.2 we shall argue that
$\hat X^2-\hat X^0$ is the quantum analogue of $z^{-1}$.  From (\ref{Xsonegbs}), its expectation value goes to infinity when $m\rightarrow \infty$.  So it follows that this is boundary limit.}

It is straightforward to generalize the notion of Killing vectors to quantized $EAdS_2$.  Denoting them by  $\hat {\cal K}^a$, their action on functions $ \hat{\cal F}$ on ${\cal A}$ is given by
\be \hat{\cal K}^a\hat{\cal F}=-\frac i\alpha[ \hat X^a,\hat{\cal F}]\;.\label{qanfKls}\ee
Just as with the Killing vectors  $ {\cal K}^a$,  $\hat {\cal K}^a$ satisfy the
$su(1,1)$ algebra
\be  [\hat {\cal K}^a,\hat {\cal K}^b]=\epsilon^{abc} \hat{\cal K}_{ c}\;,\ee
and  the action  of  $\hat {\cal K}^a$ leaves the defining relations (\ref{stpdids4X}) and (\ref{sclstpdids4X}) invariant.

In defining the Laplacian and Dirac operators on quantized $EAdS_2$ it is convenient to introduce left and right acting operators $\hat X^{La}$ and  $\hat X^{Ra}$,  respectively.  Their action on functions $\hat {\cal F}$ on quantized $EAdS_2$ is defined to be
\be \hat X^{La}\hat {\cal F}= \hat X^a\hat {\cal F}\;,\quad\quad \hat X^{Ra}\hat {\cal F}=\hat {\cal F}\hat X^a\;.\label{LRactnX}\ee
{We note that the right action is an anti-involuton, i.e., for two right acting operators $\hat A^R$ and  $\hat B^R$, $(\hat A \hat B)^R= \hat B^R\hat A^R.$  As a result there will be a sign difference between commutators of right acting operators and the corresponding commutators of left  acting operators.  In addition, right acting operators commute with a left acting operators. }  In other words,  $\hat A^R$ and  $\hat B^R$ are elements of the opposite algebra, ${\cal A}^{op}$.
Thus $\hat X^{La}$ and  $\hat X^{Ra}$ satisfy
\beqa
&& [\hat X^{La},\hat  X^{Lb}]=i\alpha\epsilon^{abc} \hat X^L_c\;,\quad\quad  [\hat X^{Ra},\hat  X^{Rb}]=-i\alpha\epsilon^{abc} \hat X^R_c\;,\quad\quad  [\hat X^{La},\hat  X^{Rb}]=0\;,\cr&&\cr
&&\hat X^{La}\hat  X_a^{L}=\hat X^{Ra}\hat  X_a^{R}=-\BI\;. \eeqa

The quantum version of the Killing vectors can  be expressed as
\be \hat{\cal K}^a=-\frac i\alpha ( \hat X^{La}- \hat X^{Ra}  ) \;.\ee
The obvious choice for the Laplacian operator on  quantized $EAdS_2$ is $ \hat \Delta=\hat{\cal K}^a\hat {\cal K}_{a}$.  It can be rewritten as
\be\hat \Delta=\,-\frac 1{\alpha^2}(  \hat X^{La}-\hat X^{Ra})(  \hat X^{L}_a-\hat X^{R}_a)\,=\,\frac 2{\alpha^2}(\hat X^{Ra}  \hat X^{L}_a+\BI)\;.\ee
So the  equation of motion for a massless  scalar field $\hat\Phi$ on quantized $EAdS_2$ {can be written in a purely algebraic way}
\be
\frac{\alpha^2}{2}\hat \Delta\hat \Phi\;=\;
\hat X_{a}\hat \Phi\hat X^{a}\,+\,\hat\Phi\;=\;0\;.
\label{XfyXeqmfy}
\ee
This equation reduces  to its commutative analogue (\ref{ctvanlgfyeq}) after requiring that the commutator of two functions $\hat{\cal F} $ and $\hat{\cal G} $ on quantized $EAdS_2$ goes to $i\alpha$ times the associated Poisson bracket in the limit $\alpha\rightarrow 0$
\be
[\hat{\cal F},\hat{\cal G}]\rightarrow
i\alpha \{\hat{\cal F},\hat{\cal G}\} \;,\quad{\rm as}\;\;\alpha\rightarrow 0\;.\label{clofopcmtr}
\ee

Before writing down the Dirac operator on  quantized $EAdS_2$, we examine the chirality operator $\hat \gamma$.    It is required to i) commute with the algebra ${\cal A}$, ii) square to unity, and iii) it should reduce to $\gamma$ in the commutative limit $\alpha\rightarrow 0$.   For i), we note that although
$\hat X^{La}$ does not commute with ${\cal A}$,  $\hat X^{Ra}$ does.
Using associativity,
$$ ( \hat X^{Ra}\hat{\cal F})\hat {\cal G}=  \hat X^{Ra}(\hat{\cal F}\hat {\cal G})=\hat {\cal F}\hat {\cal G}\hat X^a\;,$$ for any two function $\hat{\cal F}, \hat{\cal G}\in {\cal A}$,
while we get the same result from
$$ (\hat{\cal F}\hat X^{Ra})\hat{\cal  G}= \hat{\cal F}(\hat X^{Ra}\hat{\cal  G})=\hat {\cal F}\hat {\cal G}\hat X^a\;.$$
Therefore  $\hat \gamma$ should be constructed from $\hat X^{Ra}$, and since we want to recover (\ref{cvrntcrltep}) in the commutative limit, a natural candidate would be $\hat X^{Ra}\tau_a$.  However, this doesn't square to unity, and quantum corrections must be introduced to satisfy ii).
A chirality operator satisfying i-iii) was obtained by Fakhri and Imaanpur\cite{Fakhri:2003cu}:
\beqa \hat\gamma &=&\frac 1{\sqrt{1+\frac{\alpha^2}4}}\,(\hat X^{Ra}\tau_a-\frac \alpha 2\BI)\;.
\label{qntzdcvrntcrlt}\eeqa
Upon taking its square
\be \hat\gamma^2=\frac 1{{1+\frac{\alpha^2}4}}\,\Bigl(\hat X^{Ra}\hat X^{Rb}\,\tau_a\tau_b  -\alpha \hat X^{Ra}\tau_a+\frac{\alpha^2}4\Bigr)
\;=\;\BI\;.
\ee

In writing down the Dirac operator  $\hat D$ on quantized   $AdS_2$, we  should require $ a)$ that  it anti-commutes with the chirality operator $\hat\gamma $ and $b)$ that it reduces to the Dirac operator  $D$ on   $EAdS_2$ (\ref{Dpryme}) in the commutative limit.
A natural candidate  for $\hat D$ that is consistent with the commutative limit is $-i\tau_a\hat{\cal K}^a-\BI$, however this does not satisfy $a)$.
On the other hand, $a)$ is satisfied for $\hat X^{Ra}(\epsilon_{abc} \hat {\cal K}^b\tau^c+\tau_a)$
\be [\,\hat X^{Ra}(\epsilon_{abc} \hat {\cal K}^b\tau^c+\tau_a)\,,\,\hat X^{Rd}\tau_d -\frac \alpha 2\BI\,]_{ +}=0\,,\ee
 In order to satisfy $b)$, one can multiply by $ -\hat\gamma$.  The Dirac operator is then\cite{Fakhri:2003cu},
\beqa \hat D&=&-\hat \gamma\hat X^{Ra}(\epsilon_{abc} \hat {\cal K}^b\tau^c+\tau_a)=\frac i\alpha \,\epsilon_{abc}\hat \gamma\hat X^{Ra}  \hat X^{Lb}\tau^c\;.\label{dopnceads2}\eeqa
Upon expanding out the chirality operator one gets
\be
\hat D=\frac 1{\sqrt{1+\frac{\alpha^2}4}}\,\Bigl(-i\tau_a\hat{\cal K}^a-\BI+ i\alpha \hat X^{R}_a \hat{\cal K}^a-\frac {\alpha} 2 (\epsilon_{abc}\hat X^{Ra}\hat {\cal K}^b+\hat X^R_{c})\tau^c-i \hat X^R_a\tau^a\,\hat X^R_b\hat{\cal K}^b\,\Bigr)\;.\label{frsxate}
\ee
We note that $\hat X^R_{a}\hat {\cal K}^a$  acting on ${\cal A}$ is of order $\alpha$, and so to leading order in $\alpha$,  $ \hat D\sim-i\tau_a\hat{\cal K}^a-\BI$, which is consistent with the requirement $b)$.

{(\ref{dopnceads2})  is the quantization of the Dirac operator $D$ (\ref{Dpryme})  written in the covariant basis.  In Appendix B we obtain the  quantization of the Dirac operator $D$ (\ref{cmttvdo})  written in the noncovariant basis.  This is done by finding the quantum analogue of the $2\times 2 $ matrix $U$ in  (\ref{Uoprtr}), and applying it to write down the inverse of the map  (\ref{tmbtwnDs}). }

Using (\ref{dopnceads2}),
the Dirac equation for a massless particle on quantized $EAdS_2$ is
\be  \hat D\hat\Psi=\frac i\alpha \,\epsilon_{abc}\hat \gamma\hat X^{Ra}  \hat X^{Lb}\tau^c\hat\Psi=0\;,\ee
or simply
\be
\epsilon_{abc}\tau^a \hat X^{b}\hat\Psi\hat X^{c}=0\;.\label{ncDekwthm}
\ee
This equation reduces  to its commutative analogue (\ref{DeqntPBs}), after  again requiring (\ref{clofopcmtr}).
{Like the equation for the massless scalar field (\ref{XfyXeqmfy}), the field equation for a massless spinor  on quantized $EAdS_2$ has an algebraic form.  Both of these equations have exact solutions which we show in the next subsection.}

\newpage
\subsection{Exact solutions}

Here we find the exact solutions to the quantized massless Klein-Gordon (\ref{XfyXeqmfy}) and Dirac (\ref{ncDekwthm}) equations. We will start with the scalar case and then we will show how to map the scalar solutions to the spinor ones. Because the exposition is rather technical, we will structure it as the sequence of several propositions.

As the first step we want to establish the quantum analogue of the Fefferman-Graham coordinates and their relation to the quantum embedding coordinates. Classically, this relation is given in (\ref{XsnFGcrds}). Inverting this we have for the commutative Fefferman-Graham coordinates
\be\label{ztX_classic}
z = \frac{1}{X^2 - X^0}\ ,\qquad t = - z X^1 \equiv -\frac{X^1}{X^2 - X^0} \ .
\ee
As a quantization of these relations we will take the symmetric ones
\beqa\label{ztX_quantum}
\hat{z} &=& (\hat{X}^2 - \hat{X}^0)^{-1}\;,\cr &&\cr\ \hat{t} &=& -\frac{1}{2} (\hat{z} \hat{X}^1 + \hat{X}^1 \hat{z}) \equiv -\frac{1}{2} \Bigl((\hat{X}^2 - \hat{X}^0)^{-1} \hat{X}^1 + \hat{X}^1 (\hat{X}^2 - \hat{X}^0)^{-1}\Bigr)\quad
\eeqa
and show below that this is a consistent choice. Because $[\hat{z}, \hat{X}^a]= - \hat{z}[\hat{z}^{-1}, \hat{X}^a ]\hat{z}$, it is easy to find all the commutators $[\hat{z}, \hat{X}^a]$ using $\hat{z}^{-1} = \hat{X}^2 - \hat{X}^0$ and (\ref{stpdids4X})
\beqa\label{zX}
&&[\hat{z}, \hat{X}^1] = i\alpha \hat{z} \ ,\nonumber\\
&&[\hat{z}, \hat{X}^0] = - i\alpha \hat{z} \hat{X}^1 \hat{z} \equiv i\alpha \hat{z}\hat{t} - \frac{\alpha^2}{2}\hat{z}^2 \equiv i\alpha \hat{t}\hat{z} + \frac{\alpha^2}{2}\hat{z}^2 \equiv \frac{i\alpha}{2} (\hat{z}\hat{t} + \hat{t}\hat{z}) \ ,\nonumber\\
&&[\hat{z}, \hat{X}^2] = [\hat{z}, \hat{X}^0] \ .
\eeqa
From these commutators {and the definition of $\hat t$ in (\ref{ztX_quantum})} we immediately get
\be\label{xt}
[\hat{t},\hat{z}]=i\alpha \hat{z}^2 \ ,\qquad [\hat{t},\hat{z}^{-1}]= - i\alpha \ ,
\ee
i.e. the correct quantization of (\ref{tz_class}) and (\ref{zinvtcnpr}) without any quantum corrections.

To find the commutators for $\hat{t}$ and $\hat{X}^a$, we first want to find the quantum analogue of (\ref{XsnFGcrds}). It turns out that the quantization is non-trivial, compared to (\ref{ztX_quantum}). While for $\hat{X}^1$ one trivially has from (\ref{ztX_quantum})
\be\label{X1}
\hat{X}^1 = -\frac{1}{2}(\hat{z}^{-1}\hat{t} + \hat{t}\hat{z}^{-1}) \equiv - \hat{t}\hat{z}^{-1} - \frac{i\alpha}{2} \equiv - \hat{z}^{-1}\hat{t} + \frac{i\alpha}{2} \ ,
\ee
the similar relations for $\hat{X}^{0,2}$ are less trivial. Using (\ref{sclstpdids4X}) we have
\be
\hat{z}^{-1}(\hat{X}^2 + \hat{X}^0) = (\hat{X}^2 - \hat{X}^0)(\hat{X}^2 + \hat{X}^0) = -1 - (\hat{X}^1)^2 +i\alpha \hat{X}^1 \ .
\ee
Then combining this with (\ref{X1}) we get
\be
\hat{X}^2 + \hat{X}^0 = -\kappa(\alpha)^2 \hat{z} - \hat{t}\hat{z}^{-1}\hat{t}\ ,
\label{X2plsX0}\ee
where we define
\be\kappa(\alpha) :=\sqrt{1 + \frac{\alpha^2}{4}}\label{kappa}\ee
{ Notice that the same factor appears in the denominator of the quantum chirality operator (\ref{qntzdcvrntcrlt}) and  Dirac operator (\ref{dopnceads2}).
{Furthermore, from (\ref{kvsalfa}), this factor can be related to  the label $k$ of  the discrete series representation $D^+(k)$ in a simple way:  $\frac{ \kappa(\alpha)}\alpha=-k-\frac 12\;.$   (Recall, $k<-1$.) }

 Combining  (\ref{X2plsX0})} with the definition of $\hat{z}$ we finally get the full set of relations
\beqa\label{Xtz}
&&\hat{X}^1 = -\frac{1}{2}(\hat{z}^{-1}\hat{t} + \hat{t}\hat{z}^{-1}) \equiv - \hat{t}\hat{z}^{-1} - \frac{i\alpha}{2} \equiv - \hat{z}^{-1}\hat{t} + \frac{i\alpha}{2} \ ,\nonumber \\
&&\hat{X}^0 = -\frac{1}{2}(\kappa(\alpha)^2 \hat{z} + \hat{t}\hat{z}^{-1}\hat{t} + \hat{z}^{-1}) \ ,\nonumber \\
&&\hat{X}^2 = -\frac{1}{2}(\kappa(\alpha)^2 \hat{z} + \hat{t}\hat{z}^{-1}\hat{t} - \hat{z}^{-1}) \ .
\eeqa
Note the appearance of a non-trivial deformation factor, $\kappa(\alpha)^2$, which would be impossible to guess assuming some minimal quantization, as it was done in (\ref{ztX_quantum}). It is easy to verify explicitly that this factor is essential to guarantee the correct algebraic relations (\ref{stpdids4X}) {starting with the fundamental commutation relations (\ref{xt}).

Using the result (\ref{Xtz}), it is an easy exercise to find the commutators $[\hat{t},\hat{X}^a]$
\beqa\label{tX}
&&[\hat{t},\hat{X}^1] = i\alpha \hat{t} \ , \nonumber\\
&&[\hat{t},\hat{X}^0] = \frac{i\alpha}{2} (\hat{t}^2 -\kappa(\alpha)^2 \hat{z}^2 +1)\ , \nonumber\\
&&[\hat{t},\hat{X}^2] = \frac{i\alpha}{2} (\hat{t}^2 -\kappa(\alpha)^2 \hat{z}^2 -1)\ .
\eeqa

{Now we focus on solutions to the quantum scalar field equation of motion (\ref{XfyXeqmfy}), and }  formulate our first proposition:
\begin{proposition}\label{prop1}
$\hat{z}$ and $\hat t$, as defined in (\ref{ztX_quantum}), are exact solutions of the quantum scalar field equation of motion (\ref{XfyXeqmfy}).
\end{proposition}
{This is the analogue of the obvious fact that $z$ and $t$ solve the classical field equation (\ref{mslssclrfe}).  The proof of the proposition  is as follows:}
\begin{proof}
Using (\ref{sclstpdids4X}) it easy to show that (\ref{XfyXeqmfy}) can be equivalently rewritten as
\be\label{equivalent}
\hat{X}_a \hat{\Phi}\hat{X}^a + \hat{\Phi} =0\ \Leftrightarrow\ \hat{X}_a [\hat{\Phi},\hat{X}^a]=[\hat{X}_a ,\hat{\Phi}]\hat{X}^a=0 \ .
\ee
Then $\hat{X}_a [\hat{z},\hat{X}^a] = 0$ is an immediate consequence of (\ref{zX}), which shows that $\hat{z}$ is an exact solution.

For $\hat{t}$ we can proceed the same way, using (\ref{tX}), but let us prove it slightly differently by establishing some very useful criteria for a solution. Namely, if we know that some $\hat{A}$ is a solution, then $\hat{A}\hat{B}$ will be a solution for some $\hat{B}$ if and only if $\hat{X}_a \hat{A}[\hat{B},\hat{X}^a]=0$. This is trivially verified. Then one has that $\hat{t} = -\hat{z}\hat{X}^1 +\frac{i\alpha}{2}\hat{z}$ will be a solution if $\hat{X}_a \hat{z}[\hat{X}^1,\hat{X}^a]=0$, which is trivially true after some simple algebra.
\end{proof}

The classical solution (\ref{mslssln}) is given in terms of arbitrary functions of either $\z$ or $\bar{\z}$. Motivated by this and by the possibility that the quantum $\hat{z}$ ``coordinate'' could be re-scaled, as in (\ref{Xtz}), we will look for the general solutions of (\ref{XfyXeqmfy}) in the form of an arbitrary polynomial in the quantum variable $\hat{\Xi} = C\hat{z} + i\hat{t}$, where $C$ is a constant to be fixed later. First of all, in view of Proposition \ref{prop1}, $\hat{\Xi}$ is an exact solution for arbitrary $C$. The next proposition fixes this constant.
\begin{proposition}\label{prop2}
$\hat{\Xi}^2$ is an exact solution of the quantum scalar equation of motion (\ref{XfyXeqmfy}) iff $C = \pm \sqrt{1+\frac{\alpha^2}{4}}\equiv \pm\kappa(\alpha)$.
\end{proposition}
\begin{proof}
First of all, from the observation made in the proof of Proposition \ref{prop1} and from the fact that $\hat{\Xi}$ is a solution, we know that $\hat{\Xi}^2$ will be a solution if and only if $\hat{X}_a \hat{\Xi}[\hat{\Xi}, \hat{X}^a]=[\hat{X}^a ,\hat{\Xi}]\hat{\Xi}\hat{X}_a=0$. Using $[\hat{X}^a ,\hat{\Xi}]\hat{X}_a = 0$ (because $\hat{\Xi}$ is a solution), this is equivalent to
\be\label{XXiXXi}
[\hat{X}_a ,\hat{\Xi}][\hat{X}^a ,\hat{\Xi}]=0 \ .
\ee
Each commutator in (\ref{XXiXXi}) we can easily calculate from (\ref{zX}) and (\ref{tX}) giving
\beqa\label{XXi}
&&[\hat{X}^1, \hat{\Xi}] = -i\alpha \hat{\Xi} \ ,\nonumber \\
&&[\hat{X}^0, \hat{\Xi}] = -\frac{i\alpha}{2} C(\hat{z}\hat{t}+\hat{t}\hat{z}) + \frac{\alpha}{2}(\hat{t}^2 - \kappa(\alpha)^2 \hat{z}^2 +1) \ ,\nonumber \\
&&[\hat{X}^2, \hat{\Xi}] = -\frac{i\alpha}{2} C(\hat{z}\hat{t}+\hat{t}\hat{z}) + \frac{\alpha}{2}(\hat{t}^2 - \kappa(\alpha)^2 \hat{z}^2 -1) \ .
\eeqa
Plugging this in (\ref{XXiXXi}) we get
\be
0=[\hat{X}_a ,\hat{\Xi}][\hat{X}^a ,\hat{\Xi}]=\alpha^2 (\kappa(\alpha)^2 - C)\hat{z}^2
\ee
from where the result follows.
\end{proof}
So, we see that $\hat{\Xi}_\pm := \pm\kappa(\alpha)\hat{z}+i\hat t$ are the quantum analogues of the classical $\z$ and $-\bar{\z}$. Using the found value for $C$, the relations (\ref{XXi}) could be written in a more compact and closed form
\beqa\label{XXi_nice}
&&[\hat{X}^1, \hat{\Xi}] = -i\alpha \hat{\Xi} \ ,\nonumber \\
&&[\hat{X}^0, \hat{\Xi}] = -\frac{\alpha}{2} (\hat{\Xi}^2 -1) \ ,\nonumber \\
&&[\hat{X}^2, \hat{\Xi}] = -\frac{\alpha}{2} (\hat{\Xi}^2 +1) \ .
\eeqa
Now this will allow us to prove the claimed result.
\begin{proposition}\label{prop3}
An arbitrary polynomial in $\hat{\Xi}$ (either of $\hat{\Xi}_\pm$) is an exact solution of the quantum scalar equation of motion (\ref{XfyXeqmfy}).
\end{proposition}
\begin{proof}
{The proof is by induction.}  We already know that $\hat{\Xi}$ and $\hat{\Xi}^2$ are solutions (we need to start with $\hat{\Xi}^2$, because only on this level do we get a non-trivial condition on $C$). Assume that $\hat{\Xi}^{n-1}$ is also a solution. Then we know that $\hat{\Xi}^n \equiv \hat{\Xi}^{n-1}\hat{\Xi}$ will be a solution iff $\hat{X}_a \hat{\Xi}^{n-1} [\hat{\Xi},\hat{X}^a]=0$. But from (\ref{XXi_nice}) we know that $\hat{\Xi}^{n-1}$ commutes with all $[\hat{\Xi},\hat{X}^a]$, so we can write
\be
\hat{X}_a \hat{\Xi}^{n-1} [\hat{\Xi},\hat{X}^a] = \hat{X}_a [\hat{\Xi},\hat{X}^a]\hat{\Xi}^{n-1} = 0 \ ,
\ee
where we used $\hat{X}_a [\hat{\Xi},\hat{X}^a] = 0$ because $\hat{\Xi}$ is a solution.
\end{proof}

So, we see that the {general  expression  (\ref{mslssln}) for the   classical solution} has a non-trivial exact quantum counterpart.  The general solution of the quantum scalar equation of motion (\ref{XfyXeqmfy}) {can be written in terms of $\hat{\Xi}_+$ and $\hat{\Xi}_-$, or equivalently,  $\hat \Xi=\hat{\Xi}_+$ and $\hat{\Xi}^\dagger$.  It is given} by
\be
\hat{\Phi} = \hat{F}(\hat{\Xi}) + \hat{G}(\hat{\Xi}^\dagger) \ ,
\label{xctslnfrPhi}\ee
where $\hat{F}$ and $\hat{G}$ are arbitrary elements of the polynomial algebras (or of their closures) generated by $\hat{\Xi}$ and $\hat{\Xi}^\dagger$, respectively.

Now we want to ask whether the analogous {exact  solutions can} be established for the quantum Dirac equation (\ref{ncDekwthm}). The main motivation is the classical result (\ref{mslsslnprm}), which could be interpreted as arbitrary spinorial solution written in terms of arbitrary scalar solutions. The natural quantization of (\ref{mslsslnprm}) is
\be\label{psi_exact}
\hat{\Psi}=\frac 1{2}\pmatrix{1-\hat{\Xi}&-1-\hat{\Xi}^\dagger\cr 1+\hat{\Xi}&-1+\hat{\Xi}^\dagger}
\pmatrix{\hat{F}(\hat{\Xi})\cr  \hat{G}(\hat{\Xi}^\dagger)} \ .
\ee
{We want to see if it is a solution to (\ref{ncDekwthm}).}
Due to the linearity {of (\ref{ncDekwthm})}, it is enough to consider the case $\hat{F}(\hat{\Xi}) \neq 0$ and $\hat{G}(\hat{\Xi}^\dagger) = 0$. Then the general case will  trivially follow. Also by the linearity, it suffices to work only with the scalar solutions given by monomials $\hat{F}_n := \hat{\Xi}^n$. So, we want to show that
\be\label{psin_exact}
\hat{\Psi}_n :=\frac 1{2}\pmatrix{(1-\hat{\Xi})\hat{F}_n\cr (1+\hat{\Xi})\hat{F}_n } \equiv \frac 1{2}\pmatrix{\hat{F}_n - \hat{F}_{n+1}\cr \hat{F}_n + \hat{F}_{n+1} }
\ee
is an exact solution of (\ref{ncDekwthm}). The following proposition will prove to be crucial in the demonstration
\begin{proposition}\label{prop4}
Let $\hat{\Psi}'_n :=\pmatrix{\hat{F}_n\cr \hat{F}_n } $, then (\ref{psin_exact}) is given by
\be\label{psipsi}
2\alpha n \hat{\Psi}_n = \tau_a [\hat{X}^a, \hat{{\Psi}}'_n] \ ,
\ee
where $\tau_a$ are defined after (\ref{cvrntcrltep}).
\end{proposition}
\begin{proof}
From the definition of $\tau_a$ we have
\be\label{tauXPsi}
\tau_a [\hat{X}^a, \hat{{\Psi}}'_n] = \pmatrix{[\hat{X}^0 + i\hat{X}^1 +\hat{X}^2,\hat{F}_n]\cr [-\hat{X}^0 + i\hat{X}^1 -\hat{X}^2,\hat{F}_n]}\ .
\ee
Now, the relations (\ref{XXi_nice}) can be easily generalized (e.g. by recursion) to
\beqa\label{XXi_nicer}
&&[\hat{X}^1, \hat{\Xi}^n] = -in\alpha \hat{\Xi}^n \ ,\nonumber \\
&&[\hat{X}^0, \hat{\Xi}^n] = -n\frac{\alpha}{2} (\hat{\Xi}^{n+1} - \hat{\Xi}^{n-1}) \ ,\nonumber \\
&&[\hat{X}^2, \hat{\Xi}^n] = -n\frac{\alpha}{2} (\hat{\Xi}^{n+1} + \hat{\Xi}^{n-1}) \ .
\eeqa
From here we have
\be
[i\hat{X}^1 \pm(\hat{X}^0 + \hat{X}^2),\hat{\Xi}^n] = n\alpha (\hat{\Xi}^n \mp \hat{\Xi}^{n+1}) \ ,
\ee
which immediately proves the claim.
\end{proof}
Now we finally can prove the main result for the spinorial case.
\begin{proposition}\label{prop5}
The spinor (\ref{psin_exact}) (and, as a trivial consequence, the general one given in (\ref{psi_exact})) is an exact solution of the quantum Dirac equation (\ref{ncDekwthm}).
\end{proposition}
\begin{proof}
Using the result of Proposition \ref{prop4}, we need to prove that
\be\label{exact1}
\epsilon_{abc}\tau^c \hat{X}^a \tau_k [\hat{X}^k, \hat{{\Psi}}'_n] \hat{X}^b =0 \ .
\ee
Using ({\ref{towprptees}}) and
\be
\epsilon_{abc}\epsilon^c_{\ kl} = - (\eta_{ak}\eta_{bl}-\eta_{al}\eta_{bk})
\ee
we have
\be
\epsilon_{abc}\tau^c \hat{X}^a \tau_k [\hat{X}^k, \hat{{\Psi}}'_n] \hat{X}^b = i\tau_b \hat{X}_a [\hat{X}^a ,\hat{{\Psi}}'_n] \hat{X}^b - i\tau_a \hat{X}^a [\hat{X}^b ,\hat{{\Psi}}'_n] \hat{X}_b - \epsilon_{abk} \hat{X}^a [\hat{X}^k ,\hat{{\Psi}}'_n] \hat{X}^b \ .
\ee
Because each component of $\hat{{\Psi}}'_n$ is an exact solution for the scalar case, cf. (\ref{equivalent}), the first two terms are equal to zero. The remaining term we can manipulate as follows
\beqa
\epsilon_{abk} \hat{X}^a [\hat{X}^k ,\hat{{\Psi}}'_n] \hat{X}^b &=& \epsilon_{abk} \hat{X}^a \hat{X}^k \hat{{\Psi}}'_n \hat{X}^b - \epsilon_{abk} \hat{X}^a \hat{{\Psi}}'_n \hat{X}^k \hat{X}^b \nonumber \\
&=&i\alpha {X}_b \hat{{\Psi}}'_n \hat{X}^b- i\alpha {X}^a \hat{{\Psi}}'_n \hat{X}_a = 0 \ ,
\eeqa
where we used $\epsilon_{abc}\hat{X}^b \hat{X}^c =- i\alpha \hat{X}_a$.  This proves (\ref{exact1}).
\end{proof}

Consequently (\ref{psi_exact}) solves
 (\ref{ncDekwthm}).\footnote{An alternative proof is presented in Appendix B.  There we construct the quantum version $\hat U$ of the matrix $U$ in (\ref{Uoprtr}), and use it to map  the  Dirac operator to the analogue  of  the noncovariant basis.   The solutions for massless spinors in that basis are trivially obtained from the classical solutions (\ref{mslssln}) by replacing  $\zeta$  by $\Xi$, with a particular choice of ordering for the $\sqrt{\hat z}$ factor.  Upon  using $\hat U$ to map the solution back to the covariant basis, we recover  (\ref{psi_exact}).}
Therefore, the main conclusion of this section is that the exact solutions to the quantized massless Klein-Gordon (\ref{XfyXeqmfy}) and Dirac (\ref{ncDekwthm}) equations are given by essentially the same expressions as in the classical case with the non-trivial quantization of the variable $\z$ given by $\hat\Xi$.

\section{Exact boundary correlation functions}
\setcounter{equation}{0}

Boundary correlation functions are generated from the on-shell bulk action.  Before obtaining  the boundary two-point correlation functions for operators sourced by  massless scalars  and  spinors on  $EAdS_2$ on  quantized  $EAdS_2$, we  briefly recall how the correspondence works  on the classical  $EAdS_2$ bulk manifold in the first subsection.  We give the quantized version for scalars and spinors in  the following two subsections. As a result of the exact solutions of the previous section, our results are exact, and in fact, remarkably simple.

\subsection{Two-point correlators  on the boundary of  $EAdS_2$}
{We begin by reviewing the calculation of the two-point correlation function for operators sourced by the massless scalar  field on $EAdS_2$, and then for operators sourced by massless spinors.  For the latter, we use the unconventional choice of the covariant description of the Dirac operator.
 }

\subsubsection{Massless scalar field}

For the case of the  massless scalar field, one starts with the action
\beqa  S[\Phi]
&= &\frac 12\int_{ {\mathbb{R}}^2_+} dt dz\,\,\Bigl\{( \partial_z\Phi )^2 \,+\,(\partial_t\Phi)^2\Bigr\}\;.\label{clmsfa}
\eeqa
 Variations $\delta\Phi$ of $\Phi$ in  (\ref{clmsfa}) give
\beqa  \delta S[\Phi]
&= &-\int_{ {\mathbb{R}}^2_+} dt dz\,\delta\Phi\,( \partial_z^2+\partial_t^2)\Phi - \int_{ {\mathbb{R}}} dt \,(\partial_z\Phi\,\delta\Phi)\Big|_{z=0} \;.\label{forpt2}
\eeqa
Extremizing the action with Dirichlet boundary conditions yields the field equation (\ref{mslssclrfe}).
Since the equation is second order we should impose two boundary conditions to {fix the solution (\ref{sfmslssln}).} Solutions which are everywhere (and in particular at $z\rightarrow\infty$) regular can be expressed in terms of   the boundary value of the field
 $ \phi_0(t)$, using the boundary-to-bulk propagator.\cite{Witten:1998qj,Freedman:1998tz}
{Actually, in the free case, the full  boundary-to-bulk propagator is not needed for the calculation of the boundary two-point function.  The solution  (\ref{sfmslssln}) can be expanded in powers of $z$, as was done in (\ref{bhmslssln}), and  we only need to write the ${\cal O}(z) $ term, $\phi_1(t)$,   in terms of $\phi_0(t)$.}
For this we can use\footnote{This comes from the Cauchy formula.  For this define a (clockwise)  closed contour $C$ in the complex plane (parametrized by $\zeta'=z'+it'$) which goes from  $t'=-R$  to $t'=R$ along the $t'$ axis and then returns to $t'=-R$ along the semi-circle $\zeta'=Re^{i\phi}$, where $\phi$ runs from $\frac \pi 2$ to $-\frac \pi 2$.  Then
$$ f(\zeta)=-\frac{1}{2\pi i}\oint_C\frac{d\zeta'f(\zeta')}{\zeta'-\zeta}\;,\quad\quad g(\bar\zeta)=\frac{1}{2\pi i}\oint_C\frac{d\bar\zeta'g(\bar\zeta')}{\bar\zeta'-\bar\zeta}\;.
$$
Upon assuming that $f(Re^{i\phi})$ and  $g(Re^{-i\phi})$ vanish in the limit $R\rightarrow \infty$, we get that
$$   f(\zeta)=\frac{1}{2\pi }\int_{-\infty}^\infty dt'\frac{f(it')}{\zeta-it'}\;,\quad\quad g(\bar\zeta)=\frac{1}{2\pi }\int_{-\infty}^\infty dt'\frac{g(-it')}{\bar\zeta+it'}\;.
$$
(\ref{usngCf}) follows from setting $z=0$.}
\beqa f(it) =\frac 1{2\pi i}\int dt' \,\frac {f(it')}{t-t'} \,,&\quad&g(-it) =-\frac 1{2\pi i}\int dt' \,\frac {g(-it')}{t-t'}\,, \label{usngCf}\eeqa
and consequently
\beqa f'(it) =\frac 1{2\pi}\int dt'\, \frac {f(it')}{(t-t')^2}\;,
&\quad&g'(-it) =\frac 1{2\pi}\int dt'\,  \frac {g(-it')}{(t-t')^2}\;.\label{frtnf}\eeqa
From  (\ref{bhmslssln}) we then get  the limiting value of $\partial_z\Phi$
\be \partial_z\Phi\rightarrow \phi_1(t)=\frac 1{2\pi}\int dt'\,  \frac {\phi_0(t')}{(t-t')^2}\;,\qquad{\rm as}\;\;z\rightarrow 0\;.\label{phiwonfrtnf}\ee
We next substitute this back in the action (\ref{clmsfa}), which can be re-written as
 \beqa  S[\Phi]
&= &-\frac 12\int_{ {\mathbb{R}}^2_+} dt dz\,\Phi( \partial_z^2+\partial_t^2)\Phi\;-\frac 12\int_{ {\mathbb{R}}} dt \,(\Phi\partial_z\Phi)\Big|_{z=0}\;.\label{clmsfasi2}
\eeqa
Only the boundary term survives in the on-shell action.  Using (\ref{phiwonfrtnf}), the result is
\be S[\Phi]|_{\rm on-shell}=-\frac{1}{2\pi}\int_{ {\mathbb{R}}} dt \int_{ {\mathbb{R}}} dt'\,\frac{\phi_0(t)\phi_0(t') }{(t-t')^2}\;\;.
\ee
In the  $AdS/CFT$ correspondence one identifies $ S[\Phi|_{\rm on-shell}$ with the generating functional of the
$n-$point connected correlation functions for the  operator ${\cal O}$ associated with $\phi_0$.  Here, both ${\cal O}$ and  $\phi_0$ are  functions of  only $t$,
 \be <{\cal O}(t_1)\cdots{\cal O}(t_n)>=\frac{\delta^n S[\Phi]|_{\rm on-shell}}{\delta \phi_0(t_1)\cdots\delta \phi_0(t_n)}\bigg|_{\phi_0=0}\label{adscft} \ .\ee
 So the two-point function in this example is
\be <{\cal O}(t){\cal O}(t ')>\;=\;-\frac{1}{\pi}\,\frac{1 }{(t-t')^2}\label{2ptfncmtv} \ .\ee

\subsubsection{Massless spinor field}

Since we presented two different formulations for  spinors
on the classical $EAdS_2$ bulk manifold, we can proceed with two different derivations of the corresponding two-point correlator on the boundary.  These two approaches cannot be easily mapped from one to another because the map  (\ref{tmbtwnDs}) between the Dirac operators  is singular at the $z\rightarrow 0$ boundary.  Nevertheless, it can be checked that both approaches yield the same result for the boundary correlation function.  Below we work with the Dirac operator $D$ (\ref{Dpryme}) of the covariant description.

The bulk action  for massless spinors in the covariant description is given by
\be { S}_B[\overline{\Psi},\Psi]\;=\;-\int_{ {\mathbb{R}}^2_+} \frac {dzdt}{z^2}\,\overline{\Psi} D \Psi\;= \; \int_{ {\mathbb{R}}^2_+} \frac {dzdt}{z^2}\,\overline{\Psi} ( i\tau_a{\cal K}^a+\BI) \Psi \;,\label{trnsfrmddctn}\ee
	where the conjugate spinor is defined by $\overline{\Psi}={\Psi}^\dagger\sigma_3$.  With this definition one gets that $\overline{\Psi }\Psi $ is a scalar under $SU(1,1)$ transformations, $\Psi \rightarrow M\Psi  $,  where the $2\times 2$ transformation matrix satisfies $M^\dagger \sigma_3 M=\sigma_3$.\footnote{Another standard choice is $M^\dagger \sigma_2 M=\sigma_2$. In this case we should define $\overline{\Psi}={\Psi}^\dagger\sigma_2$.  This choice would lead to the appearance of    $\sigma_2$ instead of $\sigma_3$  in  the formula for the two-point correlation function below,  (\ref{wontoo5}).   }  {The action obviously leads to the field equation $D\Psi=0$, and its hermitean conjugate $\overline{\Psi}\overleftarrow{D}=0$.}

As is well known,\cite{Henningson:1998cd,Henneaux:1998ch} the action for spinor fields should be supplemented with a boundary term. {From \cite{Henneaux:1998ch}, because the field equation is first order in derivatives, we cannot fix all components of the spinor at spatial infinity (corresponding to $z\rightarrow+0$), as well as demand that the fields are everywhere regular in the bulk.  Therefore some spinor degrees of freedom at spatial infinity remain dynamical.  The   variations of the  action with regards to these boundary degrees of freedom, as well as the bulk degrees of freedom, should be consistently made to vanish.   This is not possible if the action  consists only of (\ref{trnsfrmddctn}).   It is for this reason that  boundary term should be included in the total action.   Note, that  this  situation differs for the case of a scalar field, where $\delta\Phi$ could be consistently subjected to the Dirichlet boundary conditions in (\ref{forpt2}).}

 To obtain the boundary term consider an arbitrary
 variation in  $\Psi$  in the action (\ref{trnsfrmddctn})
\beqa \delta{ S}_B[\overline{\Psi},\Psi]&=&-\int_{ {\mathbb{R}}^2_+} \frac {dzdt}{z^2}\,\overline{\Psi} D\delta \Psi\cr&&\cr&=&\Sigma_B-\lim_{z\rightarrow 0}\int_{ {\mathbb{R}}} {dt}\,\overline{\Psi}\Upsilon\delta \Psi\;,\label{thsstbvtn}\eeqa
where   $\Sigma_B$ are the terms that vanish on-shell, and $\Upsilon$ is defined by
\be \Upsilon=-\frac {i}{z}\pmatrix {t &t+i\cr -t+i& -t}\label{Upsilon}\;. \ee {The boundary contribution in  (\ref{thsstbvtn}) resulted from $\partial_z$ terms in (\ref{Dprm}).
$\Upsilon $ can also be expressed in terms of the imbedding coordinates $X^a$ and the chirality operator $\gamma$
\be \Upsilon=\epsilon_{abc} X^{a}\partial_tX^{b}\,\gamma\tau^c\;.\label{Upsilon} \ee
The asymptotic form of the solution  for $\Psi$ was given in (\ref{cantoofrthree}), while for
  $\bar\Psi$ we have
\beqa
 \overline{\Psi}(\zeta,\bar\zeta)&= & \overline{\psi_0}(t) (\sigma_3 + i t\BI )+ i  z \partial_t\Bigl(\,\overline{\psi_1}(t) (\sigma_3 + i t\BI )\Bigr)+{\cal O}(z^2)\;,\label{nxt2ldngrdr}\eeqa
where  $\psi_0$ and  $\psi_1$ were  defined in (\ref{psi01mns}). { We shall regard  $\psi_0(t)$  (and $\overline{\psi_0}(t)$) as non-dynamical, i.e.,  source fields, and so variations in (\ref{thsstbvtn}) are to be carried out with respect to the remaining boundary fields $\psi_1(t)$.}
The boundary term in  (\ref{thsstbvtn}) is removed upon adding the following  term to the bulk action (\ref{trnsfrmddctn})
  \beqa { S}_{\partial B}[\overline{\Psi},\Psi]&=&\lim_{z\rightarrow 0}\int_{\partial B} {dt}\, \overline{\psi_0}(t) (\sigma_3 + i t\BI )\Upsilon\Psi\;.\label{trnbndactn}\eeqa
The total action is then $S_{\rm total}= S_B+S_{\partial B}$.

 Only the boundary term (\ref{trnbndactn}) survives in the on-shell action.  To evaluate it we can  rewrite the asymptotic expression for the solution  (\ref{cantoofrthree})  using  (\ref{usngCf}) and (\ref{frtnf}):
\be \Psi(\zeta,\bar\zeta)\;=\;\frac 1{2\pi}\int dt'\,\Biggl(-\frac i{t-t'}(\sigma_3-it\BI)\psi_1(t')\,+\, \frac z{(t-t')^2}(\sigma_3-it'\BI)\psi_0(t')\,+\,{\cal O}(z^2)\Biggr)\;.\label{fesvnn0to}\ee
 So evaluating  $S_{\rm total}$ on-shell gives
\beqa  S_{\rm total}[\bar\Psi,\Psi]|_{\tt on-shell} &=&
\int {dt dt'}\,\frac z{2\pi(t-t')^2}\, \overline{\psi_0}(t) (\sigma_3 + i t\BI ) \Upsilon  (\sigma_3-it'\BI)
\psi_0(t')\cr&&\label{inthsznt}\eeqa
{The $z$ independent term in the integrand  of (\ref{fesvnn0to}) does not contribute to the on-shell action due to the result that $
\overline{\psi_0}(t) (\sigma_3 + i t\BI ) \Upsilon  (\sigma_3-it\BI)
\psi_1(t')$ is identically zero.
Notice  also that the factor $z$ in (\ref{inthsznt})} cancels with the $1/z$ appearing in $\Upsilon$ (\ref{Upsilon}).
Finally, use the identity
\be \overline{\psi_0}(t) (\sigma_3 + i t\BI ) \Upsilon  (\sigma_3-it'\BI)
\psi_0(t') =-\frac i z(t-t') \,\overline{\psi_0}(t) \sigma_3
\psi_0(t')\;,\label{idtinvUps}\ee
to simplify the result to
\beqa  S_{\rm total}[\bar\Psi,\Psi]|_{\tt on-shell} &=&-\frac i{2\pi}\int_{ {\mathbb{R}}} dt \int_{ {\mathbb{R}}} dt'\,
\overline{\psi_0}(t) \frac {\sigma_3}{t-t'}
\psi_0(t')\;.\eeqa
The resulting two point function correlation function on the boundary is
\be <{\cal O}_{{}_{\overline{\psi_0}}}\,(t){\cal O}_{\psi_0}(t')>\;=\;\frac{\delta^2  S_{\rm total}[\bar\Psi,\Psi]|_{\tt on-shell}}{\delta\overline{ \psi_0}(t)\delta \psi_0(t')}\bigg|_{\overline{\psi_0}=\psi_0=0}\; =\;\frac {1}{2\pi i}\,\frac{ \sigma_3}{t-t'}\;.\label{wontoo5}\ee

\subsection{Scalar two-point correlator on the boundary of  quantized $EAdS_2$}

The non-commutative generalization of the action (\ref{clmsfa})  is
\be\label{mslssclrfldactn}
\hat S[\hat \Phi]=-\frac{1}{2}{\rm Tr}\,\Bigl\{[\hat X^\mu,\hat \Phi][\hat X_\mu,\hat \Phi]\Bigr\}\ ,
\ee
where  Tr denotes a trace. The field equation (\ref{XfyXeqmfy}) follows from extremizing the action with respect to variations in $\hat \Phi$.

Here it is convenient to introduce the star product realization of the operator product.  The Moyal-Weyl star product can be employed for this purpose, provided one transforms to an appropriate  (canonical) pair of coordinates $(x,y)$,
as was discussed  in \cite{Pinzul:2017wch,deAlmeida:2019awj}.  From (\ref{zinvtcnpr}), that canonical pair can be taken to be
 $z^{-1}$ and $t$.  As  $(z^{-1},t)$ span ${\mathbb{R}}^{2}_+$, appropriate boundary conditions should be imposed at $z^{-1}=0$ (which recall is {\it not} the asymptotic boundary of $EAdS_2$), i.e., all functions should vanish sufficiently rapidly as $z\rightarrow\infty$.  Alternatively, one can employ a pair of canonical coordinates, denoted by $(x,y)$, that  span the entire plane, as was done in \cite{Pinzul:2017wch,deAlmeida:2019awj}.
Then the map back to Fefferman Graham coordinates is
\be\label{FGsymbol}
\begin{array}{l}
  t= y e^{-x}\\
  z= e^{-x}
\end{array}
\ .
\ee
In the quantum theory $x$ and $y$ get promoted to operators $\hat x$ and $\hat y$, satisfying the standard canonical commutation relation
\be [\hat x,\hat y]=i\alpha\BI\;\label{hrccrs}\ee
The quantization of the map (\ref{FGsymbol}) is obtained using symmetric ordering:
\be\label{FGNC}
\begin{array}{l}
  \hat t= \frac{1}{2}(\hat y e^{-\hat x}+e^{-\hat x}\hat y)\\
  \hat z= e^{-\hat x}
\end{array}
\;.\ee
{Then from (\ref{hrccrs}) and (\ref{FGsymbol}) we recover  the fundamental commutation relations for $\hat z$ and $\hat t$ (\ref{xt}).}

The product of operators $\hat {\cal F}$ and $\hat {\cal G}$ can be expressed in terms of the Moyal-Weyl star product of their symbols (which we also denote by $\hat {\cal F}$ and $\hat {\cal G}$, respectively), { which when expressed  in terms of canonical variables $(x,y)$ [or $(z^{-1},t)$],} is defined in the standard way
\be\label{dffstrprd}
[\hat {\cal F}\star \hat {\cal G} ](x,y) = \hat {\cal F}(x,y)\,\exp{\Bigl\{\,\frac{i\alpha}2 \,(\overleftarrow{ { \partial_ x }}\,
\overrightarrow{ {\partial_ y }}\,-\,\overleftarrow{ { \partial_ y }}\,
\overrightarrow{ {\partial_ x }}
)\,
\Bigr\}}\;\hat {\cal G}(x,y)\ .
\ee
Furthermore, the trace on the algebra, Tr, becomes $ \frac 1{\alpha^2} \int_{ {\mathbb{R}}^2} dxdy$.
Then the operator equations (\ref{FGNC}) are replaced by the corresponding equations for the for the symbols
 $\hat t$ and $\hat z$
\be
\begin{array}{l}
  \hat t= \frac{1}{2}(\hat y\star e^{-\hat x}+e^{-\hat x}\star\hat y)\\
  \hat z= e^{-\hat x}
\end{array}
\ ,
\ee
and we recover  (\ref{xt}), where $[\hat {\cal F},\hat {\cal G}]$ now denotes a star commutator $\hat {\cal F}\star\hat {\cal G}-\hat {\cal G}\star\hat {\cal F}.$
The star commutator   goes to $i\alpha$ times the corresponding Poisson bracket in  the commutative limit $\alpha\rightarrow 0$, (\ref{clofopcmtr}).   From (\ref{xt}) it follows that  given any  two functions $\hat {\cal F}$ and $\hat {\cal G}$ that are well-behaved when $z\rightarrow 0$, their star commutator  must vanish in the boundary limit, with the leading order term in $z$ agreeing with the commutative limit,
\be [\hat{\cal F},\hat{\cal G}]\rightarrow i\alpha z^2 \Bigl(\partial_t\hat{\cal F}\partial_z\hat{\cal G}-\partial_t\hat{\cal G}\partial_z\hat{\cal F}\Bigr)=i\alpha \{\hat{\cal F},\hat{\cal G}\}\qquad {\rm as}\;z\rightarrow 0\;. \label{Ctriblt}\ee
  Moreover, the star product of any two functions with a well-behaved $z\rightarrow 0$ limit reduces to the point-wise product in the boundary limit.

The expressions for the noncommutative analogues $\hat X^a$ of the embedding coordinates in term of $\hat t$ and $\hat z$, given in (\ref{Xtz}), can  be realized in terms of the star product. From them we  recover (\ref{stpdids4X}) and (\ref{sclstpdids4X}), the latter now written as
$ \hat X^a\star\hat X_{a}=-\BI\;$.
Upon  expanding  (\ref{Xtz})  in $\alpha$, it can be shown that embedding coordinates $ X^0$  and $ X^2$ only pick up a second order correction in $\alpha $  upon quantization (in fact, the  corrections to  $ X^0$  and $ X^2$ are identical), while $X^1$ is undeformed
 \be  \hat X^0=X^0-\frac {\alpha^2}8 z\;,\quad\quad  \hat X^1=X^1\;,\quad\quad \hat X^2=X^2-\frac {\alpha^2}8 z \;. \ee
We see that the quantum corrections vanish in the boundary limit $z\rightarrow 0$.  Recall that the quantum analogues  $\hat {\cal K}^a$ of Killing vectors $ {\cal K}^a$ were constructed from commutators (or star commutators) with $\hat X_a$ (\ref{qanfKls}).  Since $\hat X^a$ reduces to embedding coordinates $X^a$ in the boundary limit, and  the star commutator goes to  $i\alpha$ times the corresponding Poisson bracket in the limit (\ref{Ctriblt}), it follows that  $\hat {\cal K}^a$ must reduce to the classical Killing vectors $ {\cal K}^a$ (\ref{KaFa})  on  the boundary.  This led to the conclusion in \cite{Pinzul:2017wch, deAlmeida:2019awj} that quantized $EAdS_2$ is an asymptotically anti-de Sitter space.

From the above we can write  the  (\ref{mslssclrfldactn}) according to
\be
\hat S[\hat \Phi]=-\frac 1{2\alpha^2} \int_{ {\mathbb{R}}^2_+} \frac{dtdz}{z^2}\,\Bigl\{[\hat X^\mu,\hat \Phi]\star[\hat X_\mu,\hat \Phi]\Bigr\}\ ,
\ee
 where again $[\,,\,]$  now denotes a star commutator.  The action can be expressed as a sum of two terms, where one term vanishes on-shell, and the other term, which we denote by  $\hat S_{\partial D}[\hat \Phi]$, is only defined on the asymptotic boundary.
As was argued in \cite{Pinzul:2017wch,deAlmeida:2019awj}, the expression for the boundary action  is  identical in form to that appearing in the commutative theory (\ref{clmsfasi2}),
\be \hat S_{\partial D}[\hat \Phi]=-\frac12\lim_{z\rightarrow 0}\int dt \,\hat \Phi\partial_z  \hat\Phi\;.\label{ncbndactfii}
\ee

{The exact solution to the quantized scalar field equation is given by (\ref{xctslnfrPhi}).  Near  the boundary $z=0$ the solution behaves as
 $$ \hat\Phi(\hat\Xi,\hat\Xi^\dagger)=\phi_0(t)+\kappa(\alpha) z\, \phi_1(t) +{\cal O}(z^2)\;,$$
\be \phi_0(t)={\hat F(it)+ \hat G(-it)}\;\qquad\; \phi_1(t)=\Bigl({\hat F'(\hat\Xi)+ \hat G'(\hat\Xi^\dagger)}\Bigr)|_{z=0}\;,\label{bhmslsslnnc}\ee
the prime denoting a derivative.  Here we used $ z\, \phi_1(t)=
 z\star \phi_1(t)+{\cal O}(z^2)=\phi_1(t)\star  z+{\cal O}(z^2)$.
So, like with the classical solution, $ \hat\Phi$
 tends to $\phi_0(t)$ as $z\rightarrow 0$, while $\partial_z\hat \Phi(z,t)$ picks up a factor of $\kappa(\alpha)$:}
 \beqa \partial_z\hat \Phi(z,t)&\rightarrow&\kappa(\alpha)\, \partial_z\Phi(z,t)|_{z\rightarrow 0}\cr&&\cr
&\rightarrow&\,\frac {\kappa(\alpha)}\pi \int dt'\frac{\phi_0(t')}{(t-t')^2}
\;,\label{drvlFeyez0}\eeqa
where we used (\ref{phiwonfrtnf}).
Upon substituting into (\ref{ncbndactfii}), the result for the on-shell action is
\be S[\Phi]|_{\rm on-shell}=-\frac{{\kappa(\alpha)}}{2\pi}\int_{ {\mathbb{R}}} dt \int_{ {\mathbb{R}}} dt'\,\frac{\phi_0(t)\phi_0(t') }{(t-t')^2}\;\;.
\ee
 So the boundary two-point function for massless scalar fields on quantized $EAdS_2$ picks up an overall factor of $\kappa(\alpha)$
\be <{\cal O}(t){\cal O}(t ')>\;=\;-\frac{{\kappa(\alpha)}}{\pi}\,\frac{1 }{(t-t')^2}\label{2ptfncmtv} \ .\ee
This agrees up to ${\cal O}(\alpha^4)$ with the leading order perturbative result found  in \cite{Pinzul:2017wch,deAlmeida:2019awj}.

\subsection{Spinor two-point correlator on the boundary of  quantized $EAdS_2$}

{As in the previous subsection we can utilize the star product realization of the operator product to write down the equations of motion and the action.  Concerning the latter, }  the  bulk action for a massless spinor on quantized $EAdS_2$  can be taken to be
 \beqa &&\hat{ S}_B[\overline{\hat \Psi},\hat \Psi]=-\int_B \frac {dzdt}{z^2}\,\overline{\hat \Psi}\star \hat D \hat \Psi\cr&&\cr&&\;=- \frac i\alpha\int_B \frac {dzdt}{z^2}\,\epsilon_{abc} \overline{\hat\Psi}\star\hat\gamma\Bigl(\tau^c \hat X ^{b}\star \hat\Psi \star \hat X ^{a}\Bigr)  \cr&&\cr &&\;= - \frac i{\alpha\sqrt{1+\frac{\alpha^2}4}} \int_B \frac {dzdt}{z^2}\epsilon_{abc}\overline{\hat \Psi}\star\Bigl(\tau_d\tau^c \hat X ^{b}\star \hat \Psi \star \hat X ^{a}\star\hat X ^{d} -\frac \alpha 2\tau^c \hat X ^{b}\star \hat \Psi \star \hat X ^{a}  \Bigr)\;.\cr&&\label{ncblkactn}
 \eeqa
Like in the commutative case, the conjugate spinor is defined by $\overline{\hat \Psi}={\hat \Psi}^\dagger\sigma_3$.
Variations of $\overline{\hat \Psi} $ in the action lead to the equations of motion in the bulk (\ref{ncDekwthm}).
Using the cyclic property of the Moyal-Weyl star product, $\int dx dy\, {\cal F}\star{\cal G}\star{\cal H}=\int dx dy\, \hat{\cal H}\star\hat{\cal F}\star\hat{\cal G}\,+\,$boundary terms,   the  action (\ref{ncDekwthm}) can also be written as
$$ \hat{ S}_B[\overline{\hat \Psi},\hat \Psi]\;=\;- \, \frac i{\alpha\sqrt{1+\frac{\alpha^2}4}} \int_B \frac {dzdt}{z^2}\epsilon_{abc}\,\hat X^{a}\star \Bigl(\hat X^{d}\star\overline{\hat\Psi}\tau_d   -\frac \alpha 2 \overline{\hat\Psi}\Bigr)\star\hat X^{b}\tau^c  \star\hat\Psi\,+\,B.T. $$
\vspace{-2mm}
 $$
\;=  \frac i{\alpha\sqrt{1+\frac{\alpha^2}4}}\int_B \frac {dzdt}{z^2} {\epsilon_{abc} }\biggl(\hat X^{d}\star \hat X^{a}\star\overline{\hat \Psi} \star\hat X^{b} \tau^c\tau_d-\frac \alpha 2   \hat X^{a}\star \overline{\hat \Psi}\star \hat X^{b}\tau^c \biggr) \star \hat\Psi+B.T$$
\be
= \frac i\alpha\int_B \frac {dzdt}{z^2} \, {\epsilon_{abc} }\hat\gamma\Bigl(\hat X^{a}\star\overline{\hat\Psi} \star\hat X^{b} \tau^c\Bigr) \star\hat \Psi\,+\,B.T.\;,\qquad\qquad\qquad\qquad\qquad\qquad
 \ee
where $B.T.$ denotes boundary terms, and we used the identity
\be
\epsilon_{abc} \hat X^{a}\star\hat X^{d}\star\overline{\hat\Psi} \star \hat X^{b}\tau_d  \tau^c  =-\epsilon_{abc}\hat X^{d}\star\hat X^{a}\star\overline{\hat\Psi} \star \hat X^{b} \tau^c\tau_d+\alpha\epsilon_{abc}\hat X^{a}\star\overline{\hat\Psi} \star \hat X^{b} \tau^c\;.\label{dfcltid} \ee
Then variations of $\hat{\Psi} $ in the action lead to the equations of motion \be   {\epsilon_{abc} } \hat X^{a}\star\overline{\hat\Psi} \star \hat X^{b} \tau^c=0\label{fratefe}\;.\ee
This is the hermitean conjugate of the Dirac equation (\ref{ncDekwthm}). To show this we can use the identity
$\tau_a^\dagger=\sigma_3\tau_a\sigma_3$.

The exact solution to the Dirac equation (\ref{ncDekwthm}) { is given by (\ref{psi_exact}).}
As before, we only need the leading and next to leading behavior as $z\rightarrow 0$  to obtain the two-point correlator on the boundary.
{The expansion of (\ref{psi_exact}) near $z=0$ is
\beqa\hat \Psi(\hat\Xi,\hat \Xi^\dagger)&= & \frac{1}{2}
\left(
\begin{array}{c}
 1-i t \\1+ i t \\
\end{array}
\right) (\hat F(i t)-\hat G(-i t))\cr&&\cr&+&\frac{\kappa(\alpha)z}{2} \left(
\begin{array}{c}
-\hat F(i t)-\hat G(-i
   t)+(1-i t) \left(\hat F'(i t)-\hat G'(-i t)\right)\\
\;\;\,\hat F(i t)+\hat G(-i t)+(1+i
   t) \left(\hat F'(i t)-\hat G'(-i t)\right)\\
\end{array}
\right)+O\left(z^2\right)\,.\cr&& \label{nxt2ldngrdr}\eeqa
}
 It, along with its hermitean conjugate, can be re-expressed as
\beqa \hat \Psi(\hat\Xi,\hat \Xi^\dagger)&=&
 (\sigma_3-i t\BI )\,\psi_0(t)
- i  z\kappa(\alpha)\,\partial_t\Bigl(  (\sigma_3-i t\BI )\,\psi_1(t)\Bigr)+{\cal O}(z^2)\;,\cr &&\cr
 \overline{\hat\Psi}(\hat\Xi,\hat \Xi^\dagger)&= & \overline{\psi_0}(t) (\sigma_3 + i t\BI )+ i  z\kappa(\alpha)\, \partial_t\Bigl(\,\overline{\psi_1}(t) (\sigma_3 + i t\BI )\Bigr)+{\cal O}(z^2)\label{nxt2ldngrdr}\;,\label{toofrthree}\eeqa
where  $\psi_0$ and  $\psi_1$ are the quantum analogues of (\ref{psi01mns}).

As was true in the commutative case, the action should be supplemented with a boundary term so that variations of the total action can be consistently made to vanish.
A general variation in  $\hat\Psi$ induces the following change in the bulk action (\ref{ncblkactn}):
 \beqa \delta\hat{ S}_B&=& - \frac i\alpha\frac 1{\sqrt{1+\frac{\alpha^2}4}} \int_B \frac {dzdt}{z^2}\epsilon_{abc} \overline{\hat\Psi}\star\Bigl(\tau_d\tau^c \hat X^{b}\star \delta\hat \Psi \star \hat X^{a}\star\hat  X^{d} -\frac \alpha 2\tau^c \hat X^{b}\star\delta \hat\Psi \star \hat X^{a}  \Bigr)\cr&&\cr &=&- \frac i\alpha\,\frac {1}{\sqrt{1+\frac{\alpha^2}4}} \int_B \frac {dzdt}{z^2}\epsilon_{abc}\hat X^{a}\star \Bigl(\hat  X^{d}\star\overline{\hat \Psi}\tau_d   -\frac \alpha 2 \overline{\hat\Psi}\Bigr)\star\hat X^{b}\tau^c  \star\delta\hat\Psi\; +\; b.\,v.\;,\cr&&\label{scndln482}\eeqa
where the boundary variation is
\be b.\,v. =-  \frac i\alpha\frac 1{\sqrt{1+\frac{\alpha^2}4}} \int_B \frac {dzdt}{z^2}\epsilon_{abc} \Biggl(\Bigl[ \overline{\hat\Psi}\star\tau_d\tau^c \hat X^{b}\star \delta\hat\Psi\, ,\, \hat X^{a}\star\hat  X^{d}\Bigr]-\frac \alpha 2\Bigl[ \overline{\hat\Psi}\star\tau^c \hat X^{b}\star\delta \hat\Psi \,,\, \hat X^{a} \Bigr]\Biggr)\;. \label{dltsptlB}\ee
{The boundary term results from the integral of a star commutator, and as shown in  appendix A, such an integral is given explicitly by}
\be  \int_{D}\frac{ dzdt}{z^2}\, [\hat {\cal F},\hat {\cal G}](z,t)=-i \alpha\int{dt}\,(\partial_t \hat {\cal F}\,\hat{\cal G})|_{z=0}\;,\label{smplrslt} \ee
for any two functions $\hat {\cal F}$ and $\hat {\cal G}$ on quantized $EAdS_2$.
  Applying this to (\ref{dltsptlB}) gives
\be  b.\,v. =- \frac 1{\sqrt{1+\frac{\alpha^2}4}} \int dt\,\epsilon_{abc} \Biggl((\hat X^{a}\star\hat  X^{d})\partial_t\Bigl( \overline{\hat\Psi}\star \hat X^{b}\star\tau_d\tau^c\delta \hat\Psi\Bigr)-\frac \alpha 2\hat X^{a} \,\partial_t\Bigl( \overline{\hat\Psi}\star \hat X^{b}\star\tau^c\delta \hat\Psi\Bigr)\Biggr)\Bigg|_{z=0} \;. \ee
Using the result that  the star product reduces to the pointwise product in the boundary limit, and  that $\hat X^a$ go to the embedding coordinates $X^a$ in the boundary limit,  this simplifies to {
\beqa b.\,v.& =&- \frac 1{\sqrt{1+\frac{\alpha^2}4}} \int dt\,\epsilon_{abc} X^{a}\partial_tX^{b}\;\overline{\hat\Psi} \Bigl(X^{d} \tau_d-\frac \alpha 2\Bigr)\tau^c\delta\hat \Psi\Big|_{z=0}\cr&&\cr & =&- \frac 1{\sqrt{1+\frac{\alpha^2}4}} \int dt\,\overline{\hat\Psi} \Bigl(\BI-\frac \alpha 2\gamma\Bigr)\Upsilon\delta\hat \Psi\Big|_{z=0}\label{ncbndvr}\;, \eeqa
where $\gamma$ is the chirality operator on the classical manifold (\ref{cvrntcrltep}), and $ \Upsilon$ was defined in (\ref{Upsilon}). }
We see that the boundary term reduces to its commutative counterpart (\ref{thsstbvtn}) in the limit $\alpha\rightarrow 0$.

As in section 4.1.2, we  need to add a term to the action that will cancel the boundary variation (\ref{ncbndvr}).  For this we shall again assume that  $ {\psi_0}$ is the source field.  Then  (\ref{ncbndvr}) is eliminated upon adding
\beqa \hat{S}_{\partial B}[\overline{\hat \Psi},\hat \Psi]&= &\frac 1{\sqrt{1+\frac{\alpha^2}4}} \int dt\,\,\overline{\psi_0}(t) (\sigma_3 + i t\BI )\Bigl(\BI-\frac \alpha 2\gamma\Bigr)\Upsilon\hat\Psi\Big|_{z=0}
 \label{prencbndrem} \eeqa
to the bulk action (\ref{ncblkactn}). Upon using the identity,
\be  \overline{\psi_0}(t) (\sigma_3 + i t\BI )\gamma\Upsilon=- \overline{\psi_0}(t) \;,\ee
(\ref{prencbndrem}) simplifies to
\beqa \hat{ S}_{\partial B}[\overline{\hat \Psi},\hat \Psi]&= &\frac 1{\sqrt{1+\frac{\alpha^2}4}} \int dt\,\overline{\psi_0}(t)\biggl( (\sigma_3 + i t\BI )\Upsilon+\frac \alpha 2\biggr)\hat\Psi\Big|_{z=0}\;.
 \label{ncbndrem} \eeqa

 Once again only the boundary term survives in the total  action  $\hat S_{\rm total}= \hat S_B+\hat S_{\partial B}$ when evaluating it  on-shell.
 To evaluate it we can  rewrite the asymptotic expression for the solution for $\hat \Psi$ (\ref{nxt2ldngrdr}) using  (\ref{usngCf}) and (\ref{frtnf}):
\be \hat \Psi(\hat\Xi,\hat\Xi^\dagger)\;=\;\frac 1{2\pi}\int dt'\,\Biggl(-\frac i{t-t'}(\sigma_3-it\BI)\psi_1(t')\,+\, \frac{\kappa(\alpha) z}{(t-t')^2}(\sigma_3-it'\BI)\psi_0(t')\,+\,{\cal O}(z^2)\Biggr)\;.\label{thfesvnn0to}\ee
Substituting this result into  $\hat S_{\partial B}$, and regarding $\psi_0$ as the source field, we get
\beqa \hat  S_{\rm total}[\overline{\hat\Psi},\hat\Psi]|_{\tt on-shell} &=&\frac {\kappa(\alpha)}{\sqrt{1+\frac{\alpha^2}4}}
\int {dt dt'}\,\frac z{2\pi(t-t')^2}\, \overline{\psi_0}(t) (\sigma_3 + i t\BI ) \Upsilon  (\sigma_3-it'\BI)
\psi_0(t')\,.\cr&&\eeqa
Like in the commutative case, this term  survives  after taking the $z\rightarrow 0$ limit because $\Upsilon$ (\ref{Upsilon}) contains a factor of $1/z$, while the term that is linear in $\alpha$  in the integrand of (\ref{ncbndrem}) does not survive.
Using  the identity (\ref{idtinvUps}),
we then get
\beqa  \hat  S_{\rm total}[\overline{\hat\Psi},\hat\Psi]|_{\tt on-shell} &=&-\frac i{2\pi}\frac {\kappa(\alpha)}{\sqrt{1+\frac{\alpha^2}4}}\int_{ {\mathbb{R}}} dt \int_{ {\mathbb{R}}} dt'\,
\overline{\psi_0}(t) \frac {\sigma_3}{t-t'}
\psi_0(t')\;.\eeqa
But we saw in (\ref{kappa}) that  $ {\kappa(\alpha)}={\sqrt{1+\frac{\alpha^2}4}}$, so the on-shell action and resulting two-point function on the boundary are identical to the  commutative case to all orders in $\alpha$!

\section{Concluding remarks}
\setcounter{equation}{0}

In this paper we have continued our examination of the non-commutative $AdS/CFT$ correspondence.  Though here we restricted our attention to the case of massless fields, we have extended previous results in two directions.  First of all,  by  considering spinors,  we have gone beyond the case of scalar fields examined in \cite{Pinzul:2017wch,deAlmeida:2019awj ,Lizzi:2020cwx}.  An essential tool for the study of spinor fields on the quantized $EAdS_2$ was the map  (\ref{tmbtwnDs})
to the covariant formulation of the Dirac operator, which we were able to extend to the quantized space [cf. appendix B]. The second direction of this paper makes what we regard as a notable advance in the study of fields on quantized spaces-times. In our previous works we obtained results only to leading non-trivial order in the non-commutative parameter $\alpha$, while here we were able to solve the free scalar and spinor equations exactly, i.e., non-perturbatively in $\alpha$.  The results suggest  more generally that there may be a simple prescription for finding exact solutions on this quantum space.  It is to replace functions of the complex variables $\zeta=z+it$ and  $\bar\zeta=z-it$  with functions of the quantum operators $\hat\Xi= \kappa(\alpha)\hat{z}+i\hat t$ and $\hat \Xi^\dagger=\kappa(\alpha)\hat{z}-i\hat t$, respectively, where $\hat{z}$ and $\hat t$ are the noncommutative analogues of the Fefferman-Graham coordinates $z$ and $t$.  Thus it appears that the $z$ coordinate (which is the inverse of the radial coordinate) gets re-scaled by $\kappa(\alpha)$ upon quantization.  The $z$ coordinate is said to be associated with the energy scale of the corresponding $CFT$.  The implication then is that   the $CFT$ energy scale  gets re-scaled by the quantization of the  bulk.

We conjectured in our previous works that  boundary conformal symmetry survives the quantization of the $AdS$  bulk space-time.  By having exact results  we have now  been able to prove the conjecture  for free massless fields.  Furthermore, we have shown that the boundary two-point function is undeformed for the case of a spinor source, while it  gets re-scaled by  $\kappa(\alpha)$ for the case of a scalar source.

There are a number of extensions of this work that can be considered, some of them being crucial (in particular, item 2. below) to the study of the non-commutative correspondence.  Among them are:

\begin{enumerate}

\item  One obvious extension of our work is to obtain exact solutions for the case of { massive } scalars and spinors on quantized $EAdS_2$.  Exact results are possible and will likely involve special functions, with the domain spanned by operators $\hat\Xi$ and $\hat\Xi^\dagger$. Exact results for the corresponding boundary two-point functions can then be deduced from the asymptotic properties of these solutions, as was done here.

\item Another obvious extension is to include interactions in  scalar and spinor field theory on quantized $EAdS_2$, such as $\lambda \Phi^3$ or a Yukawa coupling.  Partial results for the $\lambda\Phi^3$ interaction were obtained in  \cite{deAlmeida:2019awj}, but because we didn't have exact solutions they required a rather involved expansion in two parameters, $ \lambda$ and the non-commutativity parameter $\alpha$.  A simpler one parameter expansion can now be pursued by exploiting  the exact solution we have found here.  It should admit leading order expressions of $n(>2)$ point correlators on the boundary. A question raised previously was whether or not the re-scaling of the boundary two-point function (in the case of scalars) could be due to a trivial field renormalization.   The study of interactions using the exact results found here should resolve this issue.

\item
Natural generalizations of quantized $EAdS_2$ were obtained in \cite{Lizzi:2020cwx}.
They correspond to the quantization of ${\mathbb{CP}}^{p,q}$.  As with $EAdS_2$, the quantization is uniquely given upon demanding that it preserves the isotropies. A construction of the Laplacian and Dirac operators on these spaces should be possible, from which one can search for exact solutions to the equations for massless scalar and spinor fields analogous to the ones found in this article.

\end{enumerate}

We plan to address these and  related topics  in forthcoming articles.

\bigskip
\bigskip

\appendice{\Large\bf ~~~Boundary terms}
\setcounter{equation}{0}

\bigskip

 The integral of the commutator of the   Moyal-Weyl star product of any two functions $\hat {\cal F}$ and $\hat{\cal G}$ on the plane  vanishes provided that the functions vanish sufficiently rapidly at infinity.  The integral does not vanish, however, if the domain $D $ (which say is two-dimensional) has a boundary $\partial D$ and the functions are unconstrained on the boundary. Below  we show that if $D$  is the bulk region of $EAdS_2$, we get the result  (\ref{smplrslt}).

Say that $D$ is some two-dimensional region with a boundary $\partial D$.  From the the definition of the  Moyal-Weyl star product (\ref{dffstrprd}), one gets
\be  \int_{D} dxdy\, [\hat {\cal F},\hat {\cal G}](x,y)=\int_{D} dxdy\,[\partial_x{\cal V}_y-\partial_y{\cal V}_x](x,y)=\int_{\partial D}\,( {\cal V}_xdx +{\cal V}_ydy)\;.\label{bndryfsp}\ee
Upon expanding  $ {\cal V}_x $ and $ {\cal V}_y $ to leading and next to leading   orders in  $\alpha$ one has
\beqa   {\cal V}_x&=&i\alpha \biggl(-\partial_x\hat {\cal F}\,\hat{\cal G} +\frac{\alpha^2}{24}\Bigl(\partial_x^3\hat{\cal F\,}\partial^2_{y}\hat{\cal G}+\partial_x\partial^2_{y}\hat{\cal F}\,\partial^2_{x}\hat{\cal G}-2\partial_x^2\partial_{y}\hat{\cal F}\,\partial_{x}\partial_{y}\hat{\cal G}\Bigr)+{\cal O}(\alpha^4)\biggr)\ , \cr &&\cr {\cal V}_y&=&i\alpha \biggl(-\partial_y\hat{\cal F}\,\hat{\cal G} +\frac{\alpha^2}{24}\Bigl(\partial_y^3\hat{\cal F}\,\partial^2_{x}\hat{\cal G}+\partial^2_x\partial_{y}\hat{\cal F}\,\partial^2_{y}\hat{\cal G}-2\partial_x\partial^2_{y}\hat{\cal F}\,\partial_{x}\partial_{y}\hat{\cal G}\Bigr)+{\cal O}(\alpha^4)\biggr)\ .\cr&& \label{vxandvy}\eeqa
Now let $D$  correspond to the bulk region of $EAdS_2$ with $\partial D$ corresponding to $z=0$, where the map between the canonical coordinates $(x,y)$ to Fefferman-Graham coordinates can be taken to be (\ref{Drbx}).  A constant $z$ slice corresponds to $x = const$. So using  $dy =
\frac{ dt}z$, we get the following integral along the asymptotic boundary
\be  \int_{D} dxdy\, [\hat{\cal F},\hat{\cal G}](x,y)=\int dt \,\frac{{\cal V}_y}z\Big|_{z=0}\;.\ee
${\cal V}_y$ is re-expressed in terms for Fefferman-Graham coordinates, using \be \partial_x=-z\partial_z -t\partial_t\;, \quad\quad \partial_y=z\partial_t\;.\ee
Assuming that $\hat{\cal F}$ and $\hat{\cal G}$, along with their derivatives, go to well defined functions of $t$ in the limit $z\rightarrow 0$, we get that the order $\alpha^3$ correction terms in ${\cal V}_y$ go like $z^3$, and so  do not contribute in the limit.  Thus only the ${\cal O}(\alpha)$  term survives in the  limit, giving the result (\ref{smplrslt}).

\bigskip
\bigskip

\appendice{\Large \bf ~~~Quantum map to noncovariant  basis}

\bigskip

In section 2 we gave the map (\ref{tmbtwnDs})  between the covariant and  noncovariant formulations of the Dirac operator on $EAdS_2$. Here we give  the quantum analogue of the map.
This will allow us to construct the Dirac operator on quantized $EAdS_2$ in the noncovariant basis.  We shall see that solutions to the Dirac equation for massless spinors are trivial in this basis.
As a check, we can then map these solutions back to the covariant basis. The result agrees with the  exact solutions   obtained in section 3,  (\ref{psi_exact}).

To obtain the quantum map we will need the quantum version of the $2\times 2$
 matrix $U$  (\ref{Uoprtr}).  For this we first introduce the left and right acting operators $\hat z^{A}$ and  $\hat t^{A}$, $A=L,R$, which act on functions $\hat{\cal F}$ on quantized $EAdS_2$ according to
\beqa
 \hat z^{L}\hat {\cal F}=\hat z\hat{\cal F}&\;,\quad\quad& \hat t^{L}\hat{\cal F}=\hat t\hat{\cal F}\;,\cr&&\cr \hat z^{R}\hat{\cal F}=\hat{\cal F}\hat z&\;,\quad\quad& \hat t^{R}\hat{\cal F}=\hat{\cal F}\hat t\ \;.\eeqa
 From the commutation relations (\ref{xt}),
 the left acting  operators $\hat t^L$ and $\hat z^L$ satisfy
\be [\hat t^L,\hat z^L]=i\alpha (\hat z^L)^2\;,\ee
while the opposite sign appears in the commutator of
$\hat t^R$ and $\hat z^R$
\be [\hat t^R,\hat z^R]=-i\alpha (\hat z^R)^{2}\;.\label{htrhzr}\ee
All other commutators between  $\hat z^{A}$ and  $\hat t^{A}$ vanish.
From (\ref{htrhzr}) it follows that
\be [\hat t^R,(\hat z^R)^{-1}]=i\alpha\BI\;.\label{cmtrtzinv}\ee
We can further introduce the right-acting version of the operator $\hat\Xi$ defined in section 3.2, along with its hermitean conjugate:
\be \hat\Xi^R=\kappa(\alpha) \hat z^R+i\hat t^R\;,\quad\quad\hat\Xi^{\dagger R}=\kappa(\alpha) \hat z^R-i\hat t^R\;.\label{XiR}\ee

We now claim that  the quantization of
 of the $2\times 2$ matrix  (\ref{Uoprtr}) is
\be  \hat U^R=\pmatrix{-\hat\Xi^R+1&-\hat\Xi^{\dagger R} -1\cr \hat\Xi^R+1&\hat\Xi^{\dagger R} -1}\frac 1{2\sqrt{\kappa(\alpha) \hat z^R}}\;.\label{hatUR}
\ee
It satisfies the analogue of the identity (\ref{prptee4U}), i.e.,
\be \hat {\tilde U}^{\dagger R}\sigma_3\hat U^R=-\sigma_3\label{opvrsnprpt}\;,\ee
where
\footnote{We can assume that  hermitean conjugation  of right acting operators $\hat A^R$ is an anti-involution.  For this define the action of
$(\hat A^R)^\dagger$   on functions $\hat{\cal F}$ on quantized $AdS_2$ by
$$ \hat{\cal F}(\hat A^R)^\dagger=(\hat A^R\hat{\cal F}^\dagger)^\dagger=(\hat{\cal F}^\dagger\hat A)^\dagger=\hat A^\dagger\hat{\cal F}\;.$$
Then for two right acting operators $\hat A^R$ and  $\hat B^R$
$$\hat {\cal F}(\hat A^R)^\dagger(\hat B^R)^\dagger=(\hat A^\dagger\hat{\cal F})(\hat B^R)^\dagger=\hat B^\dagger\hat A^\dagger\hat{\cal F}=(\hat A\hat B)^\dagger\hat{\cal F}= \hat{\cal F}((\hat A\hat B)^R)^\dagger=\hat {\cal F}((\hat B)^R(\hat A)^R)^\dagger\;.\qquad$$}
\be \hat { U}^{\dagger R}=\frac 1{2\sqrt{\kappa(\alpha) \hat z^R}}\pmatrix{-\hat\Xi^{\dagger R}+1&\hat \Xi^{\dagger R}+1\cr-\hat\Xi^{ R} -1&\hat\Xi^{ R} -1}\;.\ee
The inverse of $\hat U^R$ is
\be  ( \hat U^R)^{-1}=\frac 1{2\sqrt{\kappa(\alpha) \hat z^R}}\pmatrix{\hat\Xi^{\dagger R}-1&\hat \Xi^{\dagger R}+1\cr-\hat\Xi^{ R} -1&-\hat\Xi^{ R} +1}\;.\label{hatURinv}\ee
It is easy to check that it is the left inverse of $\hat U^R$, $ ( \hat U^R)^{-1} \hat U^R=\BI$.  To check that it is the right inverse, $ \hat U^R( \hat U^R)^{-1} =\BI$, use
\be [\hat \Xi^R,(\hat z^R)^{-1}]= -[\hat\Xi^{\dagger R},(\hat z^R)^{-1}]=-\alpha\;,\ee
which follows from (\ref{cmtrtzinv}).

From (\ref{prptee4U}), we saw that $U$ maps the classical  chirality operator $\sigma_3$ of the noncovariant basis to the chirality operator  $ \gamma $ of the covariant basis.
The analogous property holds for $\hat U$ in the quantized system.
  Upon computing  $ \hat U^R\sigma_3( \hat U^R)^{-1}$, we get
\be   \hat U^R\sigma_3( \hat U^R)^{-1}=\frac 1{\sqrt{1+\frac{\alpha^2}4}}\,(\hat X^{Ra}\tau_a-\frac \alpha 2\BI)\;,\ee
which is the chirality operator $\hat\gamma$ on quantized $EAdS_2$.  The inverse map takes $\hat\gamma$ back to $\sigma_3$,
\be  ( \hat U^R)^{-1}\hat \gamma \hat U^R=\sigma_3\;.\label{UrivgmuhUr}\ee

Now we apply the map to the Dirac operator $\hat D$, which in section 3 was formulated in the covariant basis.  We write this Dirac operator as
\be\hat D=\frac i\alpha \, \hat \gamma\Delta_c\tau^c\;,\qquad\;\;\Delta_c= \,\epsilon_{abc}\hat X^{Ra}  \hat X^{Lb}\;,\ee
where $\hat X^{Aa}$, $A=L,R$, give the left and right action of $\hat X^a$, (\ref{LRactnX}).  Using (\ref{Xtz}), we can express them in terms of  $\hat z^{A}$ and  $\hat t^{A}$ according to
\beqa
\hat X^{A0}&=&-\frac 12 \Bigl(\hat t^{A}(\hat z^{A})^{-1}\hat t^{A}+\kappa(\alpha)^2\hat z^{A}+(\hat z^{A})^{-1}\Bigr)\;,\cr&&\cr
\hat X^{A1}&=&-\frac 12 \Bigl(\hat t^{A}(\hat z^{A})^{-1}+\hat t^{A}(\hat z^{A})^{-1}\Bigr)\;,\cr&&\cr
\hat X^{A2}&=&-\frac 12 \Bigl(\hat t^{A}(\hat z^{A})^{-1}\hat t^{A}+\kappa(\alpha)^2\hat z^{A}-(\hat z^{A})^{-1}\Bigr)\;.\label{XLRitozt}
\eeqa

To write down the Dirac operator in the noncovariant basis we  apply  a similarity transformation to  $\hat D$.  We get
\be \hat D^{ ( \hat U^R)^{-1}}\; =\; ( \hat U^R)^{-1}\hat D \hat U^R\;=\;\frac i\alpha \sigma_3\, ( \hat U^R)^{-1}\Delta_c\tau^c\hat U^R\;,\label{beefftn}\ee
where we used  (\ref{UrivgmuhUr}).
Substituting (\ref{hatUR}) and (\ref{hatURinv}) into (\ref{beefftn}) gives
\be\hat D^{ ( \hat U^R)^{-1}} =\frac i{2\alpha\kappa(\alpha)} \sigma_3\,     \frac 1{\sqrt{\hat z^R}}\pmatrix{M_{11}& M_{12}\cr M_{21}&M_{22}}\frac 1{\sqrt{\hat z^R}}\;,\ee where the matrix elements $M_{ab}$ are
\beqa  M_{11}&=&\hat\Xi^{\dagger R}\Delta_+\hat\Xi^R+\Delta_- +i(\hat\Xi^{\dagger R}\Delta_1-\Delta_1\hat\Xi^R)\;,\cr&&\cr
M_{12}&= &\hat\Xi^{\dagger R}\Delta_+\hat\Xi^{\dagger R}-\Delta_- -i(\hat\Xi^{\dagger R}\Delta_1+\Delta_1\hat\Xi^{\dagger R})\;,
\cr&&\cr
M_{21}&= &-\hat \Xi^{ R}\Delta_+\hat\Xi^R+\Delta_- -i(\hat\Xi^{ R}\Delta_1+\Delta_1\hat\Xi^{ R})\;,
\cr&&\cr
M_{22}&= &-\hat\Xi^ R\Delta_+\hat\Xi^{\dagger R}-\Delta_-+i(\hat\Xi^{ R}\Delta_1-\Delta_1\hat\Xi^{\dagger R})\;,
\eeqa
and $\Delta_\pm$ and $\Delta_1$ are defined according to
\be\Delta_\pm =\Delta_0\pm\Delta_2=\epsilon_{AB}\hat X^{A\mp}\hat X^{B1}\;,\quad\quad
\Delta_1=\frac 12\epsilon_{AB}\hat X^{A+}\hat X^{B-}
\;.\ee
Here  $\hat X^{A\pm}=\hat X^{A2}\pm\hat X^{A0}$, and  $ \epsilon_{AB}$ is antisymmetric, with $ \epsilon_{LR}=1$.

The  rather involved expressions for the matrix elements can be simplified.   Starting with $M_{11}$, we can write
\beqa M_{11}&=&\;\;\hat X^{L+}\Bigl(\hat X^{R1}+\frac i2 (\hat\Xi^{\dagger R}\hat X^{R-}-\hat X^{R-}\hat\Xi^{ R})\,\Bigr)\cr&&\cr
&&+\hat X^{L-}\Bigl(\hat\Xi^{\dagger R}\hat X^{R1}\hat\Xi^{R}-\frac i2 (\hat\Xi^{\dagger R}\hat X^{R+}-\hat X^{R+}\hat\Xi^{ R})\,\Bigr)\cr&&\cr
&&-\hat X^{L1}\Bigl(\hat\Xi^{\dagger R}\hat X^{R-}\hat\Xi^{R}+\hat X^{R+}\Bigr)\;.
\eeqa
Some algebra shows that this is zero. [In fact, all three lines separately vanish.  For this we need to use (\ref{XiR}) and (\ref{XLRitozt}).]
The same result follows for
 $M_{22}$, which is expanded as
\beqa M_{22}&=&\;\;\hat X^{L+}\Bigl(-\hat X^{R1}+\frac i2 (\hat\Xi^{R}\hat X^{R-}-\hat X^{R-}\hat\Xi^{\dagger  R})\,\Bigr)\cr&&\cr
&&-\hat X^{L-}\Bigl(\hat\Xi^{R}\hat X^{R1}\hat\Xi^{\dagger R}+\frac i2 (\hat\Xi^{ R}\hat X^{R+}-\hat X^{R+}\hat\Xi^{\dagger  R})\,\Bigr)\cr&&\cr
&&+\hat X^{L1}\Bigl(\Xi^{R}\hat X^{R-}\hat\Xi^{\dagger R}+\hat X^{R+}\Bigr)\;.
\eeqa
Again,  all three lines separately vanish.  So as with the commutative Dirac operator  in the noncovariant basis (\ref{cmttvdo}), the diagonal matrix elements vanish, $M_{11}=M_{22}=0$.
Concerning the off-diagonal matrix elements, we first expand $M_{21}$:
\beqa M_{21}&=&\;\;\hat X^{L+}\Bigl(\hat X^{R1}-\frac i2 (\hat\Xi^{R}\hat X^{R-}+\hat X^{R-}\hat\Xi^{ R})\,\Bigr)\cr&&\cr
&&+\hat X^{L-}\Bigl(-\hat\Xi^{ R}\hat X^{R1}\hat\Xi^{ R}+\frac i2 (\hat\Xi^{ R}\hat X^{R+}+\hat X^{R+}\hat\Xi^{  R})\,\Bigr)\cr&&\cr
&&+\hat X^{L1}\Bigl(\hat\Xi^{R}\hat X^{R-}\hat\Xi^{ R}-\hat X^{R+}\Bigr)\;.
\eeqa
With some work this can be simplified to
\beqa M_{21}&=&\;\kappa(\alpha)\biggl(-i\hat X^{L+}-i\hat X^{L-} (\hat\Xi^{ R})^2+2\hat X^{L1}\hat\Xi^R\Bigr)\cr&&\cr
&=&i\kappa(\alpha)\Biggl(\kappa(\alpha)^2\hat z^{L}+\hat t^{L}(\hat z^{L})^{-1}\hat t^{L}-(\hat z^{L})^{-1} (\hat\Xi^{ R})^2+i\Bigl((\hat z^{L})^{-1}\hat t^{L}+\hat t^{L}(\hat z^{L})^{-1}\Bigr)\hat\Xi^{ R}\Biggr)\cr&&\cr &=&
i\kappa(\alpha)\,(\hat z^{L})^{-1}( \hat\Xi^{\dagger L}+ \hat\Xi^R-\alpha \hat z^L)(\hat\Xi^L-\hat\Xi^R )\;.
\eeqa
Finally for  $M_{12}$ we get
\beqa M_{12}&=&\;-\hat X^{L+}\Bigl(\hat X^{R1}+\frac i2 (\hat\Xi^{\dagger R}\hat X^{R-}+\hat X^{R-}\hat\Xi^{\dagger R})\,\Bigr)
\cr&&\cr
&&+\hat X^{L-}\Bigl(\hat\Xi^{\dagger R}\hat X^{R1}\hat\Xi^{\dagger R}+\frac i2 (\hat\Xi^{\dagger R}\hat X^{R+}+\hat X^{R+}\hat\Xi^{\dagger  R})\,\Bigr)\cr&&\cr
&&+\hat X^{L1}\Bigl(-\hat\Xi^{\dagger R}\hat X^{R-}\hat\Xi^{\dagger R}+\hat X^{R+}\Bigr)\;,\eeqa
which simplifies to
\beqa M_{12}
&=&\kappa(\alpha)\Bigl(-i\hat X^{L+}
-i\hat X^{L-}(\hat\Xi^{\dagger R})^2
-2\hat X^{L1}\hat\Xi^{\dagger R} \Bigr)
\cr&&\cr
&=&i\kappa(\alpha)\Biggl(\kappa(\alpha)^2\hat z^{L}+\hat t^{L}(\hat z^{L})^{-1}\hat t^{L}-(\hat z^{L})^{-1} (\hat\Xi^{\dagger R})^2-i\Bigl((\hat z^{L})^{-1}\hat t^{L}+\hat t^{L}(\hat z^{L})^{-1}\Bigr)\hat\Xi^{\dagger R}\Biggr)\cr&&\cr
 &=&
i\kappa(\alpha)\,(\hat z^{L})^{-1}(\hat \Xi^{\dagger R}+\hat \Xi^L+\alpha \hat z^L)(\hat\Xi^{\dagger L}-\hat\Xi^{\dagger R} )\;.
\eeqa
So the transformed  Dirac operator $ \hat D^{ ( \hat U^R)^{-1}}$ is
\be \frac {1}{2\alpha\hat z^{L}\sqrt{\hat z^R}}\pmatrix{0&-(\hat \Xi^{\dagger R}+ \hat\Xi^L+\alpha \hat z^L)(\hat\Xi^{\dagger L}-\hat\Xi^{\dagger R} ) \cr( \Xi^{\dagger L}+ \hat\Xi^R-\alpha \hat z^L)(\hat\Xi^L-\hat\Xi^R )& 0}\frac 1{\sqrt{\hat z^R}}\;.\label{ncDopncv}\ee  It is the quantization of the standard Dirac operator on $EAdS_2$ (\ref{cmttvdo}).

When the Dirac operator (\ref{ncDopncv}) acts on a spinor $\hat {\tilde\Psi}=\pmatrix{\hat {\tilde\Psi}_1\cr\hat {\tilde\Psi}_2}$  it gives
\be  \hat D^{ ( \hat U^R)^{-1}}\hat {\tilde\Psi}= \frac {1}{2\alpha\hat z}\pmatrix{-[\hat\Xi^\dagger,\hat {\tilde\Psi}_2\frac 1{\sqrt{\hat z}}]\,\hat\Xi^\dagger-(\hat\Xi+\alpha \hat z)\,[\hat\Xi^\dagger,\hat {\tilde\Psi}_2\frac 1{\sqrt{\hat z}}]\cr\;
[\hat\Xi,\hat {\tilde\Psi}_1\frac 1{\sqrt{\hat z}}]\,\hat\Xi\;+\,(\hat\Xi^\dagger-\alpha \hat z)\,[\hat\Xi,\hat {\tilde\Psi}_1\frac 1{\sqrt{\hat z}}]}\frac 1{\sqrt{\hat z}}\;.\ee
For the case of massless spinors, the Dirac equation implies that $[\hat \Xi,\hat {\tilde\Psi}_1\frac 1{\sqrt{\hat z}}]=[\hat\Xi^\dagger,\hat {\tilde\Psi}_2\frac 1{\sqrt{\hat z}}]=0$, and therefore we get the simple solution
\be  \hat {\tilde\Psi}=\pmatrix{F(\hat\Xi)\cr G(\hat\Xi^\dagger)}\sqrt{\hat z}\;,\label{slninncvbs}\ee
yielding  the quantum version of (\ref{mslssln}).   Once again, $\hat{F}$ and $\hat{G}$ are arbitrary elements of the polynomial algebras  generated by $\hat{\Xi}$ and $\hat{\Xi}^\dagger$, respectively.  Finally, if we act on the solution  (\ref{slninncvbs})  with $\hat U$, we get
\beqa
\hat \Psi=\hat U ^R\hat{\tilde\Psi}&=&\frac 1{2\sqrt{\kappa(\alpha)}}\,\pmatrix{-\hat\Xi^R+1&-\hat\Xi^{\dagger R} -1\cr \hat\Xi^R+1&\hat\Xi^{\dagger R} -1}\pmatrix{F(\hat\Xi)\cr G(\hat\Xi^\dagger)}\cr&&\cr
&=&\frac 1{2\sqrt{\kappa(\alpha)}}\,\pmatrix{(1-\hat\Xi)F(\hat\Xi)- (1+\hat\Xi^\dagger)G(\hat\Xi^\dagger)\cr(1+\hat\Xi)F(\hat\Xi)+ (1-\hat\Xi^\dagger)G(\hat\Xi^\dagger) }\;,
\eeqa
which agrees with the exact solution  (\ref{psi_exact}).

\end{document}